\documentclass[13pt]{article}
\usepackage{float}

\usepackage{fullpage}
\usepackage{graphicx}
\usepackage{graphicx}
\usepackage{amssymb}
\usepackage{amsmath}
\usepackage{enumitem}
\usepackage{xcolor}
\usepackage[hidelinks]{hyperref}

\numberwithin{equation}{section}

\newcommand{\Hm}[1]{\leavevmode{\marginpar{\tiny%
$\hbox to 0mm{\hspace*{-0.5mm}$\leftarrow$\hss}%
\vcenter{\vrule depth 0.1mm height 0.1mm width \the\marginparwidth}%
\hbox to
0mm{\hss$\rightarrow$\hspace*{-0.5mm}}$\\\relax\raggedright #1}}}

\newtheorem{claim}{Claim}[section]
\newtheorem{theorem}[claim]{Theorem}
\newtheorem{lemma}[claim]{Lemma}

\newtheorem{corollary}[claim]{Corollary}
\newtheorem{remark}[claim]{Remark}

\newenvironment{proof}[1][Proof]{\textsl{#1.} }{\ \rule{0.4em}{0.7em}}
\usepackage{lineno}
\usepackage{amsmath}   

\begin{document}


\title{Resonances in a Dirichlet quantum waveguide coupled to a cavity}

\author{Sylwia Kondej$^1$\footnote{Corresponding author.}\,, Nikoloz Kurtskhalia$^2$}
\date{\small  $^1$ Institute of Physics, University of Zielona G\'ora, ul.\ Szafrana 4a, 65246 Zielona G\'ora, Poland, \emph{e-mail:  s.kondej@if.uz.zgora.pl}\\ \quad \quad  \,\,\,
$^2$ School of Physics, Free University of Tbilisi, 240 Davit Aghmashenebeli Alley, Tbilisi 0159, Georgia, \emph{e-mail: nkurt22@freeuni.edu.ge}
}

\maketitle

\begin{center}
\emph{Dedicated to Professor Pavel Exner on the occasion of his 80th birthday.}
\end{center}

\begin{abstract}
We consider a Dirichlet waveguide in $\mathbb{R}^n$ ($n = 2,3$) with an attached cavity. We show that if the cavity admits a small gap, then the original embedded eigenvalues turn into resonances. The main question we address is how the size of the gap affects the resonant properties, in particular the imaginary part of the resonant pole. For example, in the case of a two dimensional waveguide with a gap of size $\varepsilon$, we show that the leading order term of the resonance behaves as $\mathcal O (\varepsilon^2)$. In the three-dimensional case, if the aperture is defined by a rectangular opening with volume proportional to $\varepsilon^2$, the resonant component behaves as $\mathcal{O}(\varepsilon^4)$.
This shows that, in the analyzed class of models, the characteristic time scale
associated with the resonances is generically of order
$\mathcal{O}((\mathrm{vol}_\varepsilon)^{-2})$, where $\mathrm{vol}_\varepsilon$
denotes the volume of the aperture inducing the resonance.
\end{abstract}

\section{Introduction}

\bigskip

In this paper, we consider a class of models involving semi-infinite straight
waveguides in $\mathbb{R}^n$, with $n = 2,3$, equipped with a resonant cavity. To ensure clarity of presentation, we begin by describing a two-dimensional planar semi-infinite waveguide $\Sigma$ of width $d_2$:
$$
\Sigma = \{(x_1, x_2 )\,:\, x_1 \in [0, \infty )\,,\,\,\,x_2 \in [\,0 , d_2] \, \}\,,
$$
bearing in mind that the results obtained in this paper are applicable to a more general situation, which will be described later.
The waveguide is equipped with a cavity of width $d_1$, situated at  the closed end,  see~Fig.~\ref{fig:enter1}.
The right-hand wall of the cavity contains a gap of size $\varepsilon > 0$; more precisely, it is defined as
$$
I_\varepsilon := \{(d_1, x_2) \,:\, x_2 \in [\,0\,,t\,] \cup [\,t+\varepsilon\,, d_2 \,] \}\,,
$$
where $t \geq 0$, and $t+\varepsilon \leq d_2$.
In particular, $
I_0 = \{(d_1, x_2) \,:\, x_2 \in [\,0\,,d_2\,]\, \}\,$ and the gap is determined by $\bar I_\varepsilon =I_0 \setminus I_\varepsilon$.
The Hamiltonian of the system is defined as the Dirichlet Laplacian $-\Delta^D_{I_\varepsilon}$ acting in $L^2(\Sigma)$ and subject to Dirichlet boundary conditions on $\partial \Sigma \cup I_\varepsilon$.
\begin{figure}[H]
    \centering
    \includegraphics[width=0.6\linewidth]{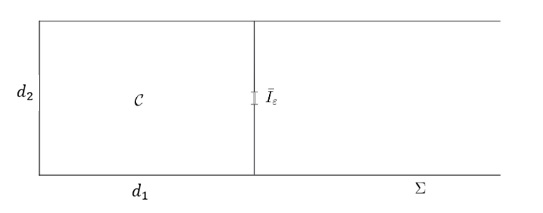 }
    \caption{Geometry of the waveguide $\Sigma$ with a cavity $\mathcal{C}$ containing  a  gap $\bar{I}\varepsilon$. }
    \label{fig:enter1}
\end{figure}
The threshold of the continuous spectrum is determined by the lowest transverse mode; therefore, the essential spectrum takes the form
\begin{equation}\label{eq-embe_intro}
\sigma_{\mathrm{ess}}(-\Delta^D_{I_\varepsilon}) = \Big[ \Big(\frac{\pi}{d_2}\Big)^2,\, \infty \Big)\,.
\end{equation}
For a closed cavity, i.e., when $\varepsilon = 0$, the trapped modes give rise to discrete energy levels
$$
\xi_{l,k} = \Big (\frac{\pi l}{d_1}\Big)^2 + \Big(\frac{\pi k}{d_2}\Big)^2,
$$
which, in fact, represent eigenvalues embedded in the essential spectrum.

Introducing a small aperture in the right hand hard wall  of the cavity enables the possibility of quantum tunnelling from the cavity into the infinite channel of the waveguide.
As a result, the previously trapped modes become metastable. It is natural to
expect that the size of the gap is related to the characteristic time scale
associated with these metastable states.
The problem can be formulated in terms of resonances, where the width of a resonance is inversely proportional to the corresponding  the characteristic time scale. The main question addressed in this paper is the following:
\emph{What is the asymptotic behavior of the resonance width as a function of $\varepsilon$?}

At this stage, we would like to emphasize that in the current paper, we also tackle the problem of a three-dimensional waveguide with the geometry illustrated in Fig.~\ref{fig:enter2}. In this case, the gap $\bar{I}_\varepsilon$ is defined by a rectangle with dimensions $\varepsilon$ and $a\varepsilon$, where $a > 0$ is a constant.
\noindent
\begin{figure}[H]
    \centering
    \includegraphics[width=0.8\linewidth]{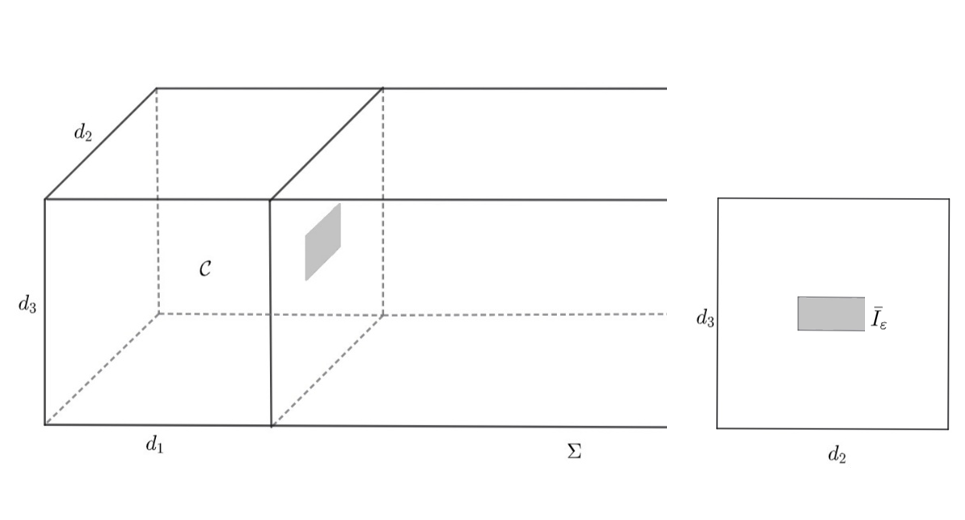}
    \caption{ Geometry of the waveguide $\Sigma $ with a box cavity $\mathcal C$ containing a rectangular gap $\bar I_\varepsilon$.}
    \label{fig:enter2}
\end{figure}

To address the problem of resonances, we formulate the spectral problem in terms of the kernel of a certain operator $K^\varepsilon(z)$, which characterizes the presence of an eigenvalue. Specifically, if
$$
\ker K^\varepsilon(z_0) \neq \emptyset,
$$
then $z_0$ is an eigenvalue of the Dirichlet Laplacian $-\Delta_{I_\varepsilon}^D$. This criterion can be considered  as a realization of the Birman–Schwinger principle in systems with a varying spatial structure, in particular those involving Dirichlet boundary conditions. Motivated by this characterization, we investigate whether there exist complex values $z$ for which
$
\ker K^\varepsilon(z) \neq \emptyset,
$
interpreted as resonance poles in the complex plane.
The main results of the paper can be formulated as follows.

\noindent \textbf{Main Results.} Assume that in a two-dimensional planar waveguide, an embedded eigenvalue $\xi_{l,k}$, given by (\ref{eq-embe_intro}), has multiplicity $\mathrm{N}$. Then, under a small perturbation induced by the opening of the gap, there exist $\mathrm{N}$ complex-valued functions $\varepsilon \mapsto z_j(\varepsilon)$, for $j = 1, \dots, \mathrm{N}$, such that
$$
\ker K^\varepsilon\big(z_j(\varepsilon)\big) \neq \emptyset,
$$
and each $z_j(\varepsilon)$ admits the asymptotic expansion
$$
z_j(\varepsilon) = \xi_{l,k} + \mu_j(\varepsilon) + i\,\nu_j(\varepsilon),
$$
where the real-valued functions $\mu_j(\varepsilon)$ and $\nu_j(\varepsilon)$ satisfy
$$
\mu_j(\varepsilon) = \mathcal{O}(\varepsilon^2) \,,\quad  \, \nu_j(\varepsilon) = \mathcal{O}(\varepsilon^2) \quad \text{as} \quad \varepsilon \to 0\,.
$$
These complex quantities $z_j(\varepsilon)$  determine the  resonance poles, with the imaginary parts describing the decay rates of the corresponding metastable states.

In three-dimensional case defined in the previous discussion the result holds true
with $
\mu_j(\varepsilon) = \mathcal{O}(\varepsilon^4) $ and  $\nu_j(\varepsilon) = \mathcal{O}(\varepsilon^4) $.

In our view, the importance of the present problem stems from two main reasons.  First, resonances play a fundamental role in determining the characteristic time scale of metastable states, which in turn influence quantum transport properties.
The results presented in this paper allow us to conclude that, for the considered
 types of waveguides, the  characteristic time scale behaves as $$\tau = \mathcal{O}(|\bar{I}_\varepsilon|^{-2})\,,$$ where $|\bar{I}_\varepsilon|$ denotes the volume of the aperture region.
A precise understanding of resonances features  enables control over the dynamical behavior of quantum particles in waveguides with the structure of $\Sigma$. This, potentially, opens the way to optimization of quantum and electronic devices where transport characteristics can be  tuned through geometric  modifications.

A second motivation for this study arises from a mathematical perspective. To rigorously description the resonance phenomena, it is necessary to develop a set of  spectral tools in  this setting. The main analytical challenge comes  from the fact that the geometry of the domain varies with the size  of the  segment $I_\varepsilon$ in the wall. This does not allow a direct application of techniques worked out in so-called soft waveguide models, cf.~\cite{Kondej2024, KondejSlipko2024}. To overcome this difficulty, we develop several auxiliary results using perturbation theory, a generalized version of the Implicit Function Theorem, analytic continuation methods for operators, and asymptotic techniques that allow for the analysis of the asymptotic behavior of resonance poles.
We believe that the mathematical framework developed here and applied to quantum systems may be of interest to both the mathematical and physical communities, as it provides insight into the relation  between geometry, spectral theory, and resonance phenomena.

To conclude, let us briefly comment on the current state of the art concerning the similar problems.
There exists a large number of papers  on Dirichlet waveguides; see, for example, \cite{EK2015} and references therein. The spectral structure of various types of waveguides has been thoroughly investigated, particularly in relation to how the geometry of a waveguide influences the existence of discrete spectrum — see, e.g., \cite{DuEx, EKK, EK2015}.
Methods for analyzing waveguides with spatially varying geometry have been developed in \cite{Post12}.

Embedded eigenvalues and resonances have been studied from various perspectives
in waveguide-like structures and in configuration spaces with interactions
localized on lower-dimensional subsets; see, for instance,
\cite{Cuenin20, Hsu16, KondejKrejcirikreson, Krejcirik23, Lipovsky2017, Posilicano22, Persson98}.
Resonant tunneling effects in waveguides weakly coupled through small geometric
openings were investigated, among others, in \cite{BorisovExnerGolovina2013},
where resonances generated by distant Dirichlet barriers were analyzed, as well
as in \cite{ChesnelNazarov21, DelitsynGrebenkov2022}, which studied almost complete
transmission through perforated screens or barriers with small apertures.
In particular, Delitsyn and Grebenkov \cite{DelitsynGrebenkov2022} demonstrated how
specific geometric configurations lead to resonant enhancement of reaction rates,
revealing close connections between diffusion dynamics, spectral properties of
the associated operators, and near-complete transmission effects.

The paper is organized as follows. In Section~\ref{sec-Description}, we formulate the problem in a two-dimensional setting and obtain preliminary spectral results. In Section~\ref{sec-Resonances}, we analyze and discuss  the existence of resonances. Section~\ref{Asymptotics} is devoted to the asymptotic analysis of the perturbative components of the resonance poles. In Section~\ref{sec-3D}, we extend the problem to a three-dimensional setting and present the main results. Finally, Section~\ref{sec-final} contains conclusions, final remarks, and open questions.

\section{Description of model and preliminary results in two-dimensional system }\label{sec-Description}

\subsection{Geometry of  planner waveguide and Hamiltonian of system}
We consider a one-sided infinite, planar waveguide with a width of $d_2$:
$$ \Sigma:= \{(x_1, x_2 )\,:\, x_1 \in [0, \infty )\,,\,\,\,x_2 \in [\,0 , d_2] \, \}\,. $$
Furthermore, we equip $\Sigma $ with a partition $I_\varepsilon \subset \Sigma$, with a aperture of the size  $\varepsilon \geq 0$, defined by
$$
I_\varepsilon := \{(d_1, x_2) \,:\, x_2 \in [\,0\,,t\,] \cup [\,t+\varepsilon\,, d_2 \,] \}\,,
$$
where $d_1>0$, $t\geq 0$ and $t+\varepsilon \leq d_2 $. For a special case $\varepsilon =0$ we have
$$
I_0= \{(d_1, x_2) \,:\, x_2 \in [\,0\,,d_2\,]\, \}\,.
$$
In the following we will use the notation $\mathcal C =\{ (x_1, x_2)\,:\, x_1\in [\,0\,,d_1\,]\,, \,\,x_2 \in [\,0\,,d_2\,] \}$.

The Hamiltonian in the waveguide system  is defined as
$$
-\Delta^D_{I_\varepsilon} \,:\, W^{2,2}_0 (\Sigma \setminus I_\varepsilon )\to L^2 (\Sigma)\,,
$$
where $ W^{2,2}_0 (\Sigma \setminus I_\varepsilon)$ denote the Sobolev space with Dirichlet boundary conditions on $ \partial \Sigma \cup I_\varepsilon$.

\subsection{Preliminary results for the special cases}

We consider two special cases: first $\varepsilon =d_2$, with $t=0$ and second $d_1>0$, $\varepsilon =0$. In the first case, we deal with the waveguide without a cavity, while in the second one, the cavity is determined by an entirely closed  hard wall located at the point $(d_1, x_2)$, $x_2\in [0, d_2]$.

The Hamiltonian corresponding to the situation $\varepsilon =d_2$ is given by $-\Delta ^D\equiv -\Delta ^D_{I_{d_2}}\,:\, W^{2,2}_0 (\Sigma )\to L^2 (\Sigma )$. Its resolvent $R(z):=(-\Delta ^D - z)^{-1}$  is an integral operator with the  kernel (Green's function) given by
\begin{equation}\label{eq-reslvent_free}
G (z;(x_1 , x_2), (y_1, y_2)):= \sum _{n=1}^\infty G_n (z;x_1,y_1 ) e^0_n(x_2) e^0_n ( y_2)\,,
\end{equation}
where \begin{equation}\label{eq-G_k}
    G_n (z;x_1,y_1) = \frac{\mathrm e ^{i\sqrt{z-(\pi n/d_2)^2}|x_1 + y_1|}}{i\sqrt{z-(\pi n/d_2)^2}}-\frac{\mathrm e ^{i\sqrt{z-(\pi n/d_2)^2}|x_1 - y_1|}}{i\sqrt{z-(\pi n/d_2)^2}}\,,\quad \Im \sqrt{z-(\pi n/d_2)^2} >0\,,
\end{equation}
and  $n\in \mathbb N$, with $\mathbb N$ defined as positive  integers; (we keep the  absolute value in the above formula although $x_1+y_1>0$), moreover
\begin{equation}
    e^0_n (\cdot ):= \sqrt{\frac{2}{d_2}}\sin \Big(\frac{\pi n }{d_2} \,\cdot \Big)\,.
\end{equation}
We define also the operator $R^\varepsilon (z)\,:\, L^2 (\Sigma )\to L^2 (\Sigma)$ with the  integral kernel
\begin{equation}\label{eq-G_k_gen}
G ^\varepsilon (z;(x_1 , x_2), (y_1, y_2)):= \sum _{n=1}^\infty G_n (z;x_1,y_1) e^\varepsilon_{n}( x_2) e^\varepsilon_{n}
( y_2)\,,
\end{equation}
where
\begin{equation}\label{eq-e_eps}
    e_{n}^{ \varepsilon} (\cdot ):= e^0_{n} (\cdot )\chi_{ I_\varepsilon } (\cdot )=\sqrt{\frac{2}{d_2}}\sin  \Big(\frac{\pi n }{d_2} \,\cdot \Big)\,\chi_{ I_\varepsilon} (\cdot )\,.
\end{equation}
and $\chi_{ I_\varepsilon} (\cdot )$ stands for the characteristic function of $ I_\varepsilon $.

To formulate the first results we define  $K^\varepsilon (z)\,: L^2 (I_\varepsilon)\to L^2 (I_\varepsilon)$ as an integral operator with the kernel
and $K^\varepsilon (z; x_2, y_2)$ given by
\begin{equation}\label{eq-def_K}
K^\varepsilon (z; x_2, y_2):=G ^\varepsilon (z;(d_1,x_2), (d_1, y_2))= \sum _{n=1}^\infty G_n (z;d_1, d_1) e_{n}^\varepsilon ( x_2) e_{n}^\varepsilon
( y_2)\,,
\end{equation}
where
$$G_n (z; d_1, d_1)= \frac{\mathrm e ^{i2d_1 \sqrt{z-(\pi n/d_2)^2}}}{i\sqrt{z-(\pi n/d_2)^2}}-\frac{1}{i\sqrt{z-(\pi n/d_2)^2}}\,,$$
cf.~(\ref{eq-G_k}) and (\ref{eq-G_k_gen}).
In fact, operator $K^\varepsilon(z)$ determines a bilateral embedding of $R^\varepsilon(z)$ to the space $L^2 (I_\varepsilon)$.


Throughout the paper, the operator $K^\varepsilon(z)$ is defined as a bounded
operator on $L^2(I_\varepsilon)$ equipped with its natural topology.
However, in the Birman--Schwinger argument below, in order to ensure
injectivity and sufficient regularity of the associated single-layer
potential, we consider $K^\varepsilon(z)$ as acting on the space
$H^{1/2}(I_\varepsilon)$.

\begin{theorem}\label{th-BS}
A number $z_0$ is an eigenvalue of $-\Delta^D_{I_\varepsilon}$ if and only if
\[
\ker K^\varepsilon(z_0)\neq\{0\},
\]
where the operator $K^\varepsilon(z)$ is considered as acting in
$H^{1/2}(I_\varepsilon)$.
Moreover,
\begin{equation}\label{eq-BS}
\dim \ker \bigl(-\Delta^D_{I_\varepsilon}-z_0\bigr)
=
\dim \ker K^\varepsilon(z_0).
\end{equation}
\end{theorem}
\begin{proof}
Assume first that $f\in\ker K^\varepsilon(z_0)\subset H^{1/2}(I_\varepsilon)$.
Define
\begin{equation}\label{eq-def_g}
g(x_1,x_2)
:=R(z_0)\big|_{\Sigma,I_\varepsilon}f
=\int_{I_\varepsilon}
G^\varepsilon\bigl(z_0;(x_1,x_2),(d_1,y_2)\bigr)f(y_2)\,\mathrm dy_2 ,
\end{equation}
By the mapping properties of the single-layer potential
(see \cite[Thm.~4.18]{McLean}),
the function $g$ belongs to $H^1_{\mathrm{loc}}(\Sigma\setminus I_\varepsilon)$
and satisfies $(-\Delta^D-z_0)g=0$ there.
Interior elliptic regularity (see \cite[Thm.~2.27]{McLean}) then yields\footnote{ Although the mapping properties of the single-layer potential are usually
formulated for densities supported on closed Lipschitz hypersurfaces,
the same argument applies verbatim on every open set
$U\Subset\Sigma\setminus I_\varepsilon$, since $\mathrm{dist}(U,I_\varepsilon)>0$
and the Green kernel is smooth in the field variable away from the support.}
$g\in H^2(\Sigma\setminus I_\varepsilon)$.
\\
For $(x_1,x_2)\in I_\varepsilon$, that is $x_1=d_1$ and
$x_2\in[0,t]\cup[t+\varepsilon,d_2]$, we have
\[
g(d_1,x_2)=(K^\varepsilon(z_0)f)(x_2)=0,
\]
and clearly $g=0$ on $\partial\Sigma$.
Therefore $g\in D(-\Delta^D_{I_\varepsilon})$.
Moreover,
\begin{equation}\label{eq-ee}
(-\Delta^D_{I_\varepsilon}-z_0)g=0
\qquad \text{in }\Sigma\setminus I_\varepsilon.
\end{equation}
Thus each $f\in\ker K^\varepsilon(z_0)$ gives rise to an eigenfunction
$g\in\ker(-\Delta^D_{I_\varepsilon}-z_0)$.

Conversely, let $g\in D(-\Delta^D_{I_\varepsilon})$ satisfy \eqref{eq-ee}.
Then $g=0$ on $\partial\Sigma\cup I_\varepsilon$.
By Green's representation formula for the Dirichlet Laplacian,
there exists $f\in H^{1/2}(I_\varepsilon)$ such that
\[
g=R(z_0)\big|_{\Sigma,I_\varepsilon}f .
\]
Restricting $g$ to $I_\varepsilon$ yields
\[
0=g(d_1,x_2)=(K^\varepsilon(z_0)f)(x_2),
\qquad x_2\in[0,t]\cup[t+\varepsilon,d_2],
\]
and hence $f\in\ker K^\varepsilon(z_0)$.

The above construction establishes a one-to-one correspondence between
$\ker K^\varepsilon(z_0)$ and
$\ker(-\Delta^D_{I_\varepsilon}-z_0)$, which proves \eqref{eq-BS}.
\end{proof}

\begin{remark} \rm{
Using the argument from the proof of Theorem~\ref{th-BS}, one concludes that
for any $z$ belonging to the resolvent set $\rho(-\Delta^D_{I_\varepsilon})$
and any function $g\in D(-\Delta^D_{I_\varepsilon})$, there exists
$f\in W^{2,2}(\Sigma)\cap W^{1,2}_0(\Sigma)$  such that
\begin{equation}\label{eq-frep}
g
=
f
-
R(z)\big|_{\Sigma,I_\varepsilon}\,
K^\varepsilon(z)^{-1}\,
\mathcal I f .
\end{equation}
Here $\mathcal I$ denotes the trace operator
\[
\mathcal I : W^{2,2}(\Sigma)\longrightarrow H^{1/2}(I_\varepsilon),
\]
and the operator $K^\varepsilon(z)$ is invertible since
$z\in\rho(-\Delta^D_{I_\varepsilon})$.
Moreover, the representation \eqref{eq-frep} holds in
$W^{2,2}(\Sigma\setminus I_\varepsilon)$.

Note that  Theorem~\ref{th-BS} represents the limiting case of the Birman--Schwinger principle for a regular or delta potential, cf.~\cite{BEG, BEKS, Gesztesy, HamsmanKrejcirik, Po1}. In the whole-space case, analogous settings are described by Kre\u{\i}n-type resolvent formulae,
linking resolvent poles to Birman--Schwinger conditions; see \cite[Sec.~7.1]{{Mantile2018}}.}
For example, if the  delta barrier is defined by a  coupling constant $\alpha > 0$ and it is localized  on $I_\varepsilon$, then the Birman--Schwinger principle takes the form
\[
\ker\left[\alpha^{-1} + K^\varepsilon(z)\right] \neq \emptyset\,.
\]
Formally taking the limit $\alpha \to \infty$, which models Dirichlet boundary conditions, we get the  equation~\eqref{eq-BS}.
\end{remark}
\subsection{Embedded eigenvalues  }

In a special case, if $\varepsilon =0$, i.e. the cavity $\mathcal C := \{(x_1, x_2)\,:\, x_1\in [0, d_1]\,, x_2\in [0, d_2]\}$ is closed, then  the Hamiltonian
$-\Delta ^D_{I_0}$
decouples
$$
-\Delta ^D_{I_0} = -\Delta ^{D} (\mathcal C) \,\dot +\, \big(-\Delta ^{D} (\Sigma \setminus \mathcal C)\big)\,,
$$
where $\Delta ^{D} (\mathcal C)$ is the Dirichlet Laplacian in $L^2 (\mathcal C)$ and $\Delta ^{D} (\Sigma \setminus \mathcal C)$ in $L^2 (\Sigma \setminus \mathcal C)$.
The operator $-\Delta ^D_{\mathcal C}$ has discrete spectrum given by
\begin{equation} \label{eq-disc_cavity}
\xi_{l,k}:=\Big( \frac{\pi l}{d_1} \Big)^2+ \Big( \frac{\pi k}{d_2} \Big)^2\,,\quad k,l \in \mathbb N \,.
\end{equation}
On the other hand, the essential spectrum   of the Hamiltonian in waveguide
has a form of the half line
and remains stable in the presence of the cavity, i.e.
\begin{equation}\label{eq-ess}
   \sigma_{\mathrm {ess}} (-\Delta^D )= \sigma_{\mathrm {ess}}  (-\Delta^D(\Sigma \setminus \mathcal C))=\sigma_{\mathrm {ess}} (-\Delta^D_{I_0}) = \sigma_{\mathrm {ess}} (-\Delta^D_{I_\varepsilon})= \Big[\,\frac{\pi^2 }{d_2^2}\,, \infty \,\Big)\,.
    \end{equation}
This means that all the numbers (\ref{eq-disc_cavity}) determine the embedded eigenvalues of $-\Delta ^D_{I_0}$.
Note, that  this result is consistent with Theorem~\ref{th-BS}. Indeed, assume that $z_0 =\xi_{l,k}$ and note  $ \mathrm e^{i2d_1\sqrt{\xi_{l,k}-(\pi k/d_2)^2}} =\mathrm e^{i2\pi l }= 1$. Therefore
$$
G_k (z_0=\xi_{l,k}; d_1, d_1 ) = 0\,.
$$
This implies  $e^0_k \in \ker K^0 (z_0=\xi_{l,k})$. The function
$$g(x_1, x_2) = \int_{I_0} G(z_0 =\xi_{l,k}; (x_1, x_2), (d_1, y_2) )\, e^0_k (y_2)\,\mathrm d y_2 \,\in L^2\, (\mathcal C)\,,$$
can be expanded as follows
\begin{align}  \nonumber
g (x_1, x_2)&=
\sum_{n=1}^\infty  \int_{I_0}G_n (z_0 =\xi_{l,k}; x_1, d_1)e^0_{n} (x_2) e^0_{n} (y_2) e^0_{k} (y_2) \mathrm d y_2\\ \nonumber &=
G_k (z_0 =\xi_{l,k}; x_1, d_1)e^0_{k} (x_2) =
d_1\Big(\frac{ \mathrm e ^{i \pi l |x_1 + d_1|/d_1}}{i\pi l}-\frac{ \mathrm e ^{i\pi l|x_1 - d_1|/d_1}}{i\pi l} \Big) e^0_{k} (x_2)\\ \nonumber
&=\frac{d_1}{i\pi l}\big(\mathrm e ^{i  \pi l (x_1 + d_1)/d_1}-\mathrm e ^{-i\pi l(x_1 - d_1)/d_1 }\big) e^0_{k} (x_2)=
\frac{2 d_1}{\pi l}\mathrm e ^{il\pi} \sin \Big(\frac{\pi l x_1}{d_1}\Big) e^0_{k} (x_2)\\ \nonumber &= \frac{d_1^{3/2}}{\pi l\sqrt{2}} \mathrm e ^{i\pi l}
e^0_l (x_1) e^0_{k} (x_2)
\,,
\end{align}
where $e^0_l (x_1) = \sqrt{\frac {2}{d_1}}  \sin \Big(\frac{\pi l x_1}{d_1}\Big)$, i.e. we use the same notation for the longitudinal and transverse cavity eigenstates, remembering that the former is associated with $d_1$, and the latter with
$d_2$. The above calculus shows that $e^0_l (\cdot ) e^0_{k} (\cdot )\chi_{\mathcal C\setminus \Sigma }$ determine the eigenfunctions  of $-\Delta^D_{I_0}$ corresponding to
$$
\xi_{l,k}=\Big(\frac{\pi l }{d_1}\Big)^2+\Big(\frac{\pi k }{d_2}\Big)^2\,,
$$
which are embedded into essential spectrum, see~(\ref{eq-ess}).

\begin{remark} {Degeneracy of the embedded eigenvalues.}
\rm{It is well known that the higher dimensional infinite  well potential can admit degenerate eigenstates and,
from the perspective of further discussion, it is important to address this issue here.

The simplest example occurs when $d_1=d_2$. In this case, any eigenvalue
$\xi_{l,k}$ is at least doubly degenerate, except for the diagonal case
$l=k$, corresponding to the pair of eigenfunctions
$e_l^0(x_1)e_k^0(x_2)$ and $e_k^0(x_1)e_l^0(x_2)$.
Of course, higher degrees of degeneracy are also possible. For example, if for a pair $(l,k)$ there exists another pair $(l',k')$ such that
\begin{equation}\label{eq-deg}
l^2 + k^2 = l'^2 + k'^2\,,
\end{equation}
then the corresponding energy level is degenerate. An example of such a case is $(l,k) = (7,1)$ and $(l',k') = (5,5)$, since
$$
7^2 + 1 = 5^2 + 5^2 = 50\,.
$$
Therefore, the energy level corresponding to $\frac{\pi^2}{d_1^2}  50$ has a degeneracy of $3$. In the case of the pair $ (1,18)$ there exist two more pairs  that satisfy~(\ref{eq-deg}); namely, $(6,17)$ and $(10,15)$, i.e. the state corresponding to $\frac{\pi^2}{d_1^2}325$ has  a degeneracy of $6$.}
\end{remark}

\section{Resonances caused by an aperture in cavity}\label{sec-Resonances}
In fact, Theorem~\ref{th-BS} shows that the problem of determining the eigenvalues of the operator \( -\Delta^D_{I_\varepsilon} \), i.e., the real poles of its resolvent, can be reformulated in terms of the kernel of \( K^\varepsilon(z) \). In this section, we take this analysis one step further and ask whether the operator \( K^\varepsilon(z) \) admits an analytic continuation to the second Riemann sheet, and whether there exists a point \( z \) with \( \Im z \leq 0 \) such that \( \ker K^\varepsilon(z) \neq \emptyset \).

The operator $K^\varepsilon (z)$, see (\ref{eq-def_K}), can written as
 \begin{equation} \label{eq-defK_ext}
 K^\varepsilon (z) = \sum_{n=1}^\infty G_n (z;d_1, d_1) P_n^\varepsilon\,:\, L^2 (I_\varepsilon) \to L^2 (I_\varepsilon)\,,
 \end{equation}
where  $P_n^\varepsilon = (e^\varepsilon_n, \cdot )_{L^2 (I_\varepsilon)} e^\varepsilon_n$.
Although \( K^\varepsilon(z) \) acts on the space \( L^2(I_\varepsilon) \), however it can be naturally extended to the space \( L^2(I_0) \) by determining  the functions \( \{e^\varepsilon_k\}_{k=1}^\infty \), cf.~(\ref{eq-e_eps}), as elements of \( L^2(I_0) \).
Of course, the functions \( e^\varepsilon_k \) do not form a basis in \( L^2(I_0) \).

In the following discussion, we will also use various modifications of \( e^0_k \).  Taking the opportunity, before proceeding with the analysis, we would like to clarify some related  notions:
\begin{equation} \label{eq-eaux}
e_k^\varepsilon= e_k^0\,\,\chi_{I_\varepsilon}\,,\quad \breve e_k^\varepsilon:= e_k^0\,\,\chi_{\bar I_\varepsilon}\,, \quad k\in \mathbb N\,.
\end{equation}

\subsection{Analytic continuation of $K^\varepsilon (z)$}

Our aim is to investigate $\ker K^\varepsilon(z)$ for values of $z$ close to the original embedded eigenvalue $\xi_{l,k}$. For this aim, we introduce the notation $\mathcal{B}(\xi_{l,k})$ to denote a small neighborhood of $\xi_{l,k}$, and we assume that $z \in \mathcal{B}(\xi_{l,k})$, i.e., $z = \xi_{l,k} + \delta$ for some $\delta \in \mathbb{C}$ with $|\delta|$ sufficiently small.
To build the analytic continuation, we distinguish two cases. Namely, for the components in (\ref{eq-defK_ext}) labeled by $n$ such that
$$\Big(\frac{\pi n }{d_2}\Big)^2 <\xi_{l,k}= \Big(\frac{\pi l }{d_1}\Big)^2+\Big(\frac{\pi k}{d_2}\Big)^2 \,,
$$ we construct the analytic continuation of $\mathcal{B} (\xi_{l,k}) \ni z \mapsto \sqrt{z - (\pi n/d_2)^2}$ to the lower second sheet through $[(\pi n/d_2)^2\,, \,\infty)$. Then, $\Im \sqrt{z - (\pi n/d_2)^2} < 0$ when $\delta$ moves to the lower half-plane.

On the other hand, for $n$ such that $(\pi n /d_2)^2 > \xi_{l,k} $, the function $\mathcal{B} (\xi_{l,k}) \ni z \mapsto \sqrt{z - (\pi n/d_2)^2}$ is analytic for both positive and negative $\Im \,\delta$ and does not require  a special  analytic continuation.

 If \( (\pi n /d_2)^2 = \xi_{l,k} \) and \( \Re \,\delta > 0 \), then we proceed as in the first case, i.e. in this situation, \( \Im \sqrt{z - (\pi n /d_2)^2} <0 \), and if \( \Re \,\delta < 0 \), then the square root expression leaves on the first Riemann sheet, i.e. \( \Im \sqrt{z - (\pi n /d_2)^2} >0 \). Note that in this case, \( G_n(z; d_1, d_1) = \mathcal{O}(\delta) \), and if \( \delta = 0 \), then \( G_n(z; d_1, d_1) = 0 \). The above construction provides the  the analytic continuation of (\ref{eq-defK_ext}). We apply these constructions to
\[
G_n (z; d_1,d_1) = \frac{\mathrm{e}^{i2d_1\sqrt{z - (\pi n/d_2)^2}}}{i\sqrt{z - (\pi n/d_2)^2}} - \frac{1}{i\sqrt{z - (\pi n/d_2)^2}}\,,
\]
where we use the same notation for the analytic continuation to the lower half-plane.

Assume that $z\in \mathcal B (\xi_{l,k})$ and employ  a special decomposition that will play an essential role in further discussion:
\begin{equation} \label{eq-K_decomposition}
K^\varepsilon (z)= K^0(z) +H^\varepsilon(z)\,,
\end{equation}
where
\begin{equation}
    \label{eq-decomH}
K^0 (z)=\sum_{n=1}^\infty G_n (z;d_1,d_1)P_n^0 \,,\quad H^\varepsilon (z) := \sum_{n=1}^\infty G_n  (z) (P^\varepsilon _{n} - P^0_{n})\,,
\end{equation}
and $P^0_n = (e^0_n\,, \cdot )_{L^2 (I_0)} e^0_n$
and $P^\varepsilon_n = (e^\varepsilon _n\,, \cdot )_{L^2 (I_0)} e^\varepsilon_n$, i.e. all the operators contributing (\ref{eq-decomH}) are defined in the space $L^2 (I_0)$; see discussion after formula (\ref{eq-defK_ext}).

\begin{lemma} \label{le-norm_1} 
The operator $ H^\varepsilon (z)$ is bounded and its norm admits the following asymptotics
\begin{equation} \label{eq-asymH}
\| H^\varepsilon(z) \|_{L^2(I_0) \to L^2(I_0)} = \mathcal{O}(\varepsilon^{1/2}) \,,
\end{equation}
that is uniform for all $z \in \mathcal{B}(\xi_{k,l})$.
\end{lemma}
\begin{proof}
All norms and scalar  products in the following proof will be understood in the space $L^2(I_0)$. For simplicity, we omit explicit reference to the space whenever there is no  confusion.

Note that employing the decomposition $e^0_n = e^\varepsilon_n +\breve e^\varepsilon_n $, see~(\ref{eq-eaux}), to
 \( P^0_{n}  = ( e^0_n\,,\cdot  )\,e^0_n \) yields
\begin{align}
     P^\varepsilon_n f-P_n^0 f= - (\breve{e}_n^\varepsilon\,, f)\, e_n^\varepsilon - (e_n^\varepsilon\,, f)\, \breve{e}_n^\varepsilon - (\breve{e}_n^\varepsilon\,,f)\, \breve{e}_n^\varepsilon\,.
\end{align}
Applying this to  \(( f\,,H^\varepsilon(z)f) \), we get
\begin{equation}
\begin{split}\label{eq-auxH}
|( f\,,H^\varepsilon(z)f) |\leq & \sum_{n=1}^\infty | G_n(z) |\,\Big[
 |(f\,,\breve{e}_n^\varepsilon )\, (e_n^\varepsilon \,,f)|\\&+ |(f\,,e_n^\varepsilon )\, (\breve{e}_n^\varepsilon\,,f)|
 + |(f\,,\breve{e}_n^\varepsilon )\, (\breve e_n^\varepsilon\,,f) |
\,\Big]\,.
\end{split}
\end{equation}
To proceed further, we begin by estimating the first of the three terms in the above bound, namely
\begin{equation}\label{eq-defT}
T: = \sum_{n=1}^\infty  |G_n(z)| \left| (f,\breve e_n^\varepsilon\,)\,( {e}_n^\varepsilon\,,f) \right|\,.
\end{equation}
Applying the Cauchy--Schwartz inequality we get
\begin{equation}\label{eq-CS1}
|(f\,,\breve {e}_n^\varepsilon )| \leq  \int_{I_0 \setminus I_\varepsilon} |\, f \,\breve{e}_n^\varepsilon \,\,| \mathrm{d}x \leq \varepsilon^{1/2} \Big( \int_{I_0 \setminus I_\varepsilon} |f \,\breve{e}_n^\varepsilon|^2\,\mathrm{ d}x \Big)^{1/2} \leq
\varepsilon^{1/2} \|f\|\,,
\end{equation}
where, in the last estimate,  we used the fact  that
$\sup |\breve{e}_n^\varepsilon | =1$.

Note that for $z\in \mathcal B(\xi_{l,k})$ and $\xi_{l,k}\neq (\pi n/d_2)^2$ we can estimate
\begin{equation}\label{eq-Gnest}
|G_n (z)|=  \frac{|\mathrm e^{i2d_1\sqrt {z-(\pi n/d_2)^2}}-1|}{|\sqrt{z-(\pi n/d_2)^2}|}\leq C \frac{1}{n}\,,
\end{equation}
where  the constant $C$ depends on $l,k, d_2$.\footnote{In the following discussion, we will use the symbol
$C$ to denote a constant, which may vary from line to line. The dependence of such constants on parameters will be explicitly indicated, and any uniformity will be stated.}
On the other hand, if there exists $n$ such that $\xi_{l,k}= (\pi n/d_2)^2$ then we have asymptotics
\( |G_n(z)| = \mathcal O (\delta ) \), where and $\delta =-\xi_{l,k}+z$. The latter can occur for one component $n$ only, therefore (\ref{eq-Gnest}) holds for any $n\in \mathbb N$.

Combining (\ref{eq-CS1}) and (\ref{eq-Gnest})
leads to
\begin{equation} \label{eq-aux5}
\begin{split}
T \leq C \,\varepsilon^{1/2} \,\|f\| \sum_{n=1}^\infty  \frac{1}{n} |(e_n^\varepsilon, f)| \leq C'\,\varepsilon^{1/2} \,\|f\| \,
\Big(\sum_{n=1}^\infty |(f,e_n^\varepsilon)|^2
\Big)^{1/2}\,,
\end{split}
\end{equation}
where $C'=C\big(\sum_{n=1}^\infty  \frac{1}{n^2} \big)^{1/2}$ and we again used the Cauchy--Schwartz inequality. Note that
$$
\sum_{n=1}^\infty |(f,e_n^\varepsilon)|^2 \leq \sum_{n=1}^\infty |(f,e_n^0)|^2 = \|f\|^2\,.
$$
Applying this to (\ref{eq-aux5}) we obtain
$$
T \leq C \,\varepsilon^{1/2} \,\|f\|^2\,.
$$
The remaining components of (\ref{eq-auxH}) can be estimated by the repeating the argument.
This yields  the asymptotics
\begin{equation}
|(f\,,(H^\varepsilon(z)f)_{L^2 (I_0)}| =\mathcal O ( \varepsilon^{1/2})\,,
\end{equation}
that is uniform with  respect to $f$ and, consequently, $\|H^\varepsilon(z)\|_{L^2 (I_0)\to L^2 (I_0)} = \mathcal O (\varepsilon^{1/2})$.
\end{proof}

\bigskip
\begin{remark}\rm{
For fixed $z$ (away from the thresholds $z=(\pi n/d_2)^2$) we have
\[
K^0(z)e_n^0 = \lambda_n(z)e_n^0,\qquad
\lambda_n(z):=G_n(z;d_1,d_1).
\]
Relying on (\ref{eq-Gnest}) we conclude that
 $K^0(z)$ is compact, since it is diagonal in  $\{e_n^0\}_{n=1}^\infty$ with eigenvalues $\to0$. Therefore every nonzero eigenvalue (or any finite set of eigenvalues) is isolated in
$\sigma(K^0(z))$, see \cite[Ch.~III, \S 6]{Kato}.

Fix $n$ such that $\lambda_n(z)\neq0$ and assume first that $\lambda_n(z)$ is simple.
Then $\delta_n(z):=\mathrm{dist}(\lambda_n(z),\sigma(K^0(z))\setminus\{\lambda_n(z)\})>0$.
Choosing $\upsilon \in(0,\delta_n(z)/2)$ and $\Gamma_n:=\{\lambda:\ |\lambda-\lambda_n(z)|=\upsilon\}$
gives a contour $\Gamma_n\subset\rho(K^0(z))$ enclosing only $\lambda_n(z)$.
Since $K^0(z)$ is normal,
\begin{equation} \label{eq-isolation}
\|(K^0(z)-\lambda)^{-1}\|
=\frac{1}{\mathrm{dist}(\lambda,\sigma(K^0(z)))}
\le \frac{1}{\upsilon},\qquad \lambda\in\Gamma_n\,,
\end{equation}
\cite[Ch.~III, \S 6]{Kato}.
In the case of a finite degeneracy, one encloses the whole finite cluster and the corresponding
Riesz projector equals the sum of the involved rank-one projectors.
}\end{remark}

\begin{lemma}\label{lem:proj-stability}
Let $I_0\subset\mathbb R$ and consider the family of bounded operators
\[
K^\varepsilon(z)=K^0(z)+H^\varepsilon(z)\qquad\text{on }L^2(I_0),
\]
defined by (\ref{eq-K_decomposition}) and (\ref{eq-decomH}) for $z\in\mathcal B(\xi_{l,k})$ and $\varepsilon\ge 0$.

Then the following statements hold for all $z\in\mathcal B(\xi_{l,k})$ and all
$0<\varepsilon\le\varepsilon_1$ for some $\varepsilon_1\in(0,\varepsilon_0]$:

\begin{enumerate}
\item[\textnormal{(1)}] $\Gamma_n$ lies in the resolvent set of $K^\varepsilon(z)$, and the Riesz projectors
\[
P_n^{K^\varepsilon}(z)
= -\frac{1}{2\pi i}\int_{\Gamma_n} (K^\varepsilon(z)-\lambda)^{-1}\,d\lambda,
\qquad
P_n^0 (z)= P_n^0
= -\frac{1}{2\pi i}\int_{\Gamma_n} (K^0(z)-\lambda)^{-1}\,d\lambda,
\]
are well-defined.

\item[\textnormal{(2)}] The projectors satisfy the estimate
\begin{equation}\label{eq:proj-difference}
\big\|P_n^{K^\varepsilon}(z)-P_n^0\big\|_{L^2(I_0)\to L^2(I_0)}
\le
\frac{|\Gamma_n|}{2\pi}
\frac{\|H^\varepsilon(z)\|_{L^2(I_0)\to L^2(I_0)}}%
{\upsilon\big(\upsilon-\|H^\varepsilon(z)\|_{L^2(I_0)\to L^2(I_0)}\big)}.
\end{equation}

\item[\textnormal{(3)}] In particular,
\begin{equation}\label{eq:proj-Oeps}
\big\|P_n^{K^\varepsilon}(z)-P_n^0\big\|_{L^2(I_0)\to L^2(I_0)}
= \mathcal O(\varepsilon^{1/2})
\qquad\text{as }\varepsilon\to 0.
\end{equation}
\end{enumerate}
\end{lemma}
\begin{proof}
Fix $z$ and write $K^0:=K^0(z)$ and $H^\varepsilon:=H^\varepsilon(z)$.
For $\lambda\in\Gamma_n$ we have
\[
K^\varepsilon-\lambda = (\mathrm I+H^\varepsilon (K^0-\lambda)^{-1})(K^0-\lambda).
\]
By \eqref{eq-isolation} we have
\[
\sup_{\lambda\in\Gamma_n}
\|(K^0-\lambda)^{-1}\|_{L^2(I_0)\to L^2(I_0)}
\le \frac1{\upsilon}.
\]
and consequently
\[
\|H^\varepsilon (K^0-\lambda)^{-1}\|_{L^2(I_0)\to L^2(I_0)}
\le \frac{\|H^\varepsilon\|_{L^2(I_0)\to L^2(I_0)}}{\upsilon}.
\]
Therefore for $\varepsilon$ small enough the right-hand side is smaller than $1$,
cf.~(\ref{eq-asymH}), and $\mathrm I+H^\varepsilon(K^0-\lambda)^{-1}$ is invertible. Moreover,
\[
(K^\varepsilon-\lambda)^{-1}
= (K^0-\lambda)^{-1}\big(\mathrm I+H^\varepsilon(K^0-\lambda)^{-1}\big)^{-1}.
\]
Using the Neumann series,
\[
\big\|(\mathrm I+H^\varepsilon(K^0-\lambda)^{-1})^{-1}\big\|_{L^2(I_0)\to L^2(I_0)}
\le \frac{\upsilon}{\upsilon-\|H^\varepsilon\|_{L^2(I_0)\to L^2(I_0)}},
\]
and therefore
\begin{equation}\label{eq:res-eps-bound}
\|(K^\varepsilon-\lambda)^{-1}\|_{L^2(I_0)\to L^2(I_0)}
\le \frac{1}{\upsilon-\|H^\varepsilon\|_{L^2(I_0)\to L^2(I_0)}}\,.
\end{equation}
The resolvent identity gives
\[
(K^\varepsilon-\lambda)^{-1}-(K^0-\lambda)^{-1}
= -(K^\varepsilon-\lambda)^{-1}H^\varepsilon(K^0-\lambda)^{-1},
\]
hence
using \eqref{eq-isolation} and \eqref{eq:res-eps-bound},
\[
\big\|(K^\varepsilon-\lambda)^{-1}-(K^0-\lambda)^{-1}\big\|_{L^2(I_0)\to L^2(I_0)}
\le
\frac{\|H^\varepsilon\|_{L^2(I_0)\to L^2(I_0)}}{\upsilon(\upsilon-\|H^\varepsilon\|_{L^2(I_0)\to L^2(I_0)})}\,.
\]
Integrating over $\Gamma_n$ proves \eqref{eq:proj-difference}, and
using \eqref{eq-asymH} yields \eqref{eq:proj-Oeps}.
\end{proof}

\begin{corollary}\label{cor:ef-asymp}
Assume that the spectral cluster of $K^0(z)$ inside $\Gamma_n$
consists of a simple eigenvalue, so that $\operatorname{rank}P_n^0=1$.
Let $e_n^0$ be the corresponding normalized eigenfunction.
For $\varepsilon>0$ sufficiently small, $\operatorname{rank}P_n^{K^\varepsilon}(z)=1$ as well.
Let $e_n^{K^\varepsilon}(z)$ denote the normalized eigenfunction spanning
$\operatorname{Ran}P_n^{K^\varepsilon}(z)$, chosen so that
\begin{equation}\label{eq-efasymp}
\big(e_n^{K^\varepsilon}(z),e_n^0\big)_{L^2(I_0)}\in\mathbb R_+\,.
\end{equation}
Then there exist $o_\varepsilon=o(1)$ and a vector $\eta_n^\varepsilon(z)\perp e_n^0$ such that
\begin{equation}\label{eq:ef-expansion}
e_n^{K^\varepsilon}(z)=(1+o_\varepsilon)e_n^0+\eta_n^\varepsilon(z),
\end{equation}
and
\begin{equation}\label{eq:eta-bound}
\|\eta_n^\varepsilon(z)\|_{L^2(I_0)}=\mathcal O(\varepsilon^{1/2}),
\qquad \text{as }\varepsilon\to0.
\end{equation}
\end{corollary}
\begin{proof}
Define
\[
\alpha_\varepsilon
:=\big(e_n^{K^\varepsilon}(z),e_n^0\big)_{L^2(I_0)}\in\mathbb R_+,
\qquad
\eta_n^\varepsilon(z)
:=e_n^{K^\varepsilon}(z)-\alpha_\varepsilon e_n^0 .
\]
By construction, $\eta_n^\varepsilon(z)\perp e_n^0$.
Since $\|e_n^{K^\varepsilon}(z)\|_{L^2(I_0)}=\|e_n^0\|_{L^2(I_0)}=1$, we have the orthogonal decomposition
\begin{equation}\label{eq:orth-decomp}
1=\|e_n^{K^\varepsilon}(z)\|_{L^2(I_0)}^2
=\alpha_\varepsilon^2+\|\eta_n^\varepsilon(z)\|_{L^2(I_0)}^2 .
\end{equation}

Recall that, with our convention on the inner product,
\[
P_n^0 u=(e_n^0,u)_{L^2(I_0)}\,e_n^0,
\qquad
P_n^{K^\varepsilon}(z)u=(e_n^{K^\varepsilon}(z),u)_{L^2(I_0)}\,e_n^{K^\varepsilon}(z).
\]
Therefore,
\[
(P_n^{K^\varepsilon}(z)-P_n^0)e_n^0
=(e_n^{K^\varepsilon}(z),e_n^0)_{L^2(I_0)}\,e_n^{K^\varepsilon}(z)-e_n^0
=\alpha_\varepsilon e_n^{K^\varepsilon}(z)-e_n^0.
\]
Taking the $L^2(I_0)$–norm and using Lemma~\ref{lem:proj-stability}, we obtain
\begin{equation} \label{eq-convproj-asymp}
\|\alpha_\varepsilon e_n^{K^\varepsilon}(z)-e_n^0\|_{L^2(I_0)}
=\|(P_n^{K^\varepsilon}(z)-P_n^0e_n^0\|_{L^2(I_0)}
\le \|P_n^{K^\varepsilon}(z)-P_n(z)\|_{L^2(I_0)\to L^2(I_0)}
=\mathcal O(\varepsilon^{1/2})\,.
\end{equation}

On the other hand, using \eqref{eq:orth-decomp},
\[
\|\alpha_\varepsilon e_n^{K^\varepsilon}(z)-e_n^0\|_{L^2(I_0)}^2
=\alpha_\varepsilon^2+1-2\alpha_\varepsilon^2
=1-\alpha_\varepsilon^2
=\|\eta_n^\varepsilon(z)\|_{L^2(I_0)}^2.
\]
Hence
\begin{equation}\label{eq-normeta}
\|\eta_n^\varepsilon(z)\|_{L^2(I_0)}
=\|(P_n^{K^\varepsilon}(z)-P_n^0\,)e_n^0\|_{L^2(I_0)}
=\mathcal O(\varepsilon^{1/2}),
\end{equation}
which proves \eqref{eq:eta-bound}. Finally, since $\alpha_\varepsilon\in[0,1]$ and $o_\varepsilon
=1-\alpha_\varepsilon$, we have
\[
|o_\varepsilon|
=|1-\alpha_\varepsilon|
\le 1-\alpha_\varepsilon^2
=\|\eta_n^\varepsilon(z)\|_{L^2(I_0)}^2
=\mathcal O(\varepsilon),
\]
and therefore $o_\varepsilon=o(1)$ as $\varepsilon\to0$.
This completes the proof.
\end{proof}


With these results, the problem of
$\ker K^\varepsilon (z)$ can be transferred to the problem of solving the following scalar equation
\begin{equation} \label{eq-spectral}
    \zeta_n (z, \varepsilon) = 0 \,,\quad \mathrm{where }\quad
    \zeta_n (z, \varepsilon):= \big (e^{K^\varepsilon}_n(z) \,, K^\varepsilon (z) e^{{K^\varepsilon}}_n(z) \big)_{L^2 (I_0)}\,.%
\end{equation}

\begin{remark} \label{re:extension}\rm{
Although the parameter $\varepsilon$ is physically restricted to $\varepsilon>0$,
in the analysis below we consider a continuous extension of the function
$\zeta_k(z,\varepsilon)$ to $\varepsilon$ in a neighborhood of $0$ in $\mathbb{R}$.
This extension is used solely as a technical device in order to apply the
implicit function theorem; the resulting solution branch is subsequently
restricted to $\varepsilon>0$.}
\end{remark}

The following theorem establishes the existence of a solution to equation (\ref{eq-spectral}).

\begin{lemma} \label{le-IFT}
There exists $\varepsilon_0>0$ and a uniquely defined, continuous function
$\varepsilon \mapsto z (\varepsilon)$, $\varepsilon \in [0, \varepsilon_0 )$ such that the equation
\begin{equation} \label{eq-exi}
 \zeta_k (z (\varepsilon), \varepsilon) = 0
\end{equation}
holds.
\end{lemma}
\begin{proof}
The first observation shows that
\begin{equation} \label{eq-obs}
\zeta_k(\xi_{l,k}, 0) = \left( e^0_k, K^0(\xi_{l,k}) e^0_k \right)_{L^2(I_0)} = 0\,.
\end{equation}

Furthermore, employing the derivative formula
\begin{equation}
\frac{d}{dz} G_n(z; d_1,d_1) = \frac{d_1 \, e^{2i d_1 \sqrt{z - ((\pi n/d_2)^2}}}{z - (\pi n/d_2)^2}
+ \frac{i ( e^{2i a \sqrt{z -  (\pi n/d_2)^2}} - 1 )}{2 ( z - (\pi n/d_2)^2)^{3/2}},
\end{equation}
and the form of $ K^0(z) $, we state that
\begin{equation}
\frac{d}{dz} \zeta_k(z, \varepsilon)\big|_{z = \xi_{l,k}, \, \varepsilon = 0} =
\Bigg[
\frac{d_1 \, e^{2i d_1 \sqrt{\xi_{l,k} - (\pi n/d_2) ^2}}}{\xi_{l,k} - ( \pi n/d_2 )^2}
+ \frac{i ( e^{2i d_1 \sqrt{\xi_{l,k} - (\pi n/d_2)^2}} - 1 )}{2 ( \xi_{l,k} - (\pi n/d_2)^2 )^{3/2}}
\Bigg]
\left( e_k^0, e_k^0 \right)_{L^2(I_0)} = \frac{d_1^3}{\pi^2}\,,
\end{equation}
i.e.
\begin{equation} \label{eq-neq}
\frac{d}{dz} \zeta_k(z, \varepsilon)\big|_{z = \xi_{l,k}, \, \varepsilon = 0} \neq 0\,.
\end{equation}

Furthermore, repeating the arguments from the proof of Lemma~\ref{le-norm_1},
one can show that  \H{BACK}
\begin{equation}
\big| (f, (K^\varepsilon(z) - K^{\varepsilon'}(z)) f)_{L^2 (I_0)} \big| \leq C |\varepsilon - \varepsilon'|^{1/2} \|f\|^2_{L^2(I_0)}
\end{equation}
for any \( f \in L^2(I_0) \). This implies that \( \varepsilon \mapsto \zeta_k(z, \varepsilon) \) is a continuous function. Relying on the fact that \( z \mapsto K^\varepsilon(z) \) is an analytic operator-valued function for \( z \in \mathcal{B}(\xi_{l,k}) \), we can conclude that \( z \mapsto \zeta_k(z, \varepsilon) \) is analytic in \( \mathcal{B}(\xi_{l,k}) \) for any  \( \varepsilon \).

We extend the function $\zeta_k(z,\varepsilon)$ to
$\varepsilon\in(-\varepsilon_0,\varepsilon_0)$, see~Remark~\ref{re:extension}. Combining the  $\varepsilon$-continuity and $z$-analyticity with~\eqref{eq-obs} and~\eqref{eq-neq}, and applying the Continuous Implicit Function Theorem, (see, for example~~\cite[Thm.~2]{Hurwicz}) we state that there exists a unique continuous function \(\varepsilon
\mapsto  z(\varepsilon) \) such that
\[
\zeta_k(z(\varepsilon), \varepsilon) = 0\,,
\]
for $\varepsilon$ small enough. This completes the proof.
\end{proof}

\section{Asymptotic analysis of resonances} \label{Asymptotics}

In this section, we investigate the asymptotic behavior of the function \(\varepsilon \mapsto z(\varepsilon)\), which satisfies equation~(\ref{eq-exi}). To derive this  asymptotics, we introduce an auxiliary function \(\delta(\varepsilon)\), defined such that
\begin{equation}\label{eq-specsolution}
z(\varepsilon) = \xi_{l,k}+\delta (\varepsilon )\,,
\end{equation}
and analyze The corrections follows from the fact that $o_\varepsilon$ is selected real
\begin{align}\nonumber
\zeta_k (z, \varepsilon)=&    \big( e^{K^\varepsilon}_k (\xi_{l,k}+\delta (\varepsilon ))\,,K^\varepsilon (z(\varepsilon ) )e^{K^\varepsilon}_k (\xi_{l,k}+\delta (\varepsilon ))\big)_{L^2 (I_0)} \\ \nonumber =&
    (1+o_\varepsilon^2)\big(e^{0}_k\,,K^0 (z(\varepsilon ))e^{0}_k  \big)_{L^2 (I_0)}
 +
   (1+ o_\varepsilon)\big( e^0_k\,,K^0 (z (\varepsilon ))\eta^\varepsilon_k (z(\varepsilon)) \big)_{L^2(I_0)} \\ \nonumber & +(1+ o_\varepsilon) \big( \eta^\varepsilon_k (z(\varepsilon))\,,K^0(z (\varepsilon ) )e^0_k\big)_{L^2 (I_0)}  +\big( \eta^\varepsilon_k (z(\varepsilon))\,, K^0(z (\varepsilon ) )\eta^\varepsilon_k (z(\varepsilon))\big)_{L^2 (I_0)}   \\ \nonumber &+
 \big( e^{K^\varepsilon}_k (z(\varepsilon))\,,H^\varepsilon (z(\varepsilon ) )e^{K^\varepsilon}_k (z(\varepsilon ))\big)_{L^2 (I_0)} \,.
\end{align}
In the above, we used the fact that $o_\varepsilon$ is chosen to be real.
Using the fact that $K^0(z (\varepsilon ) ) e^0_k$ remains in the space spanned by $e^0_k$ and $\eta^\varepsilon_k (z(\varepsilon))$ is orthogonal to $e_k^0 $ we come to the conclusion
\begin{equation}\label{eq-asymzeta}
\zeta_k (z, \varepsilon)= (1+o_\varepsilon^2)\big( e^{0}_k\,,K^0 (z (\varepsilon ) )e^{0}_k  \big)_{L^2 (I_0)} + J_\varepsilon \,,
\end{equation}
where 
\begin{equation}\label{eq-J}P
J_\varepsilon:=\big( \eta^\varepsilon_k (z(\varepsilon))\,, K^0(z (\varepsilon ) )\eta^\varepsilon_k (z(\varepsilon))\big)_{L^2 (I_0)}+
 \big (e^{K^\varepsilon}_k (z(\varepsilon))\,,H^\varepsilon (z(\varepsilon ) ) e^{K^\varepsilon}_k (z(\varepsilon ))\big)_{L^2 (I_0)} \,,
\end{equation}
which, in view of (\ref{eq-efasymp}), can be developed as
\begin{equation} \label{eq-asypJ}
    J_\varepsilon= (1+o_\varepsilon^2)\big(e^0_k \,,H^\varepsilon (z(\varepsilon ) ) e^0_k\big)_{L^2 (I_0)}+ 
   (1+o _\varepsilon)
    \big(e_k^0\,, H^\varepsilon(z (\varepsilon ) )\eta^\varepsilon_k (z(\varepsilon)) \big)_{L^2 (I_0)}
 \,,
\end{equation}
where we  used again orthogonality of $\eta_k^\varepsilon$ and $e_k^0$, the fact that $\xi_k (z (\varepsilon), \varepsilon)$ is an eigenvalue of $K^\varepsilon (z)$ to $e_k^{K^\varepsilon}(z)$ which $\xi_k (z, \varepsilon)= 0$ for (\ref{eq-specsolution}).
Employing  the asymptotics (\ref{eq-asymH}) and (\ref{eq-normeta}),
we state
$$
J_\varepsilon =\mathcal O (\varepsilon^{1/2})\,.
$$
However, in the following discussion, we refine the asymptotic behavior of \( J_\varepsilon \) by analyzing each term in equation~(\ref{eq-asypJ}). The first term is estimated in the lemma below.

\begin{lemma} \label{le-norm} For  $z\in \mathcal B (\xi_{l,k})$ the following asymptotics
\begin{equation}\label{eq-Hasymp}
\big|\big( e_k^0\,, ( H^\varepsilon (z)e_k^0 \big)_{L^2 (I_0)} \big|=\mathcal O (\varepsilon^2)\,,
\end{equation}
holds.
\end{lemma}
\begin{proof}
    Using the decomposition (\ref{eq-decomH}) we have 
 \begin{align}
       \big(e_m^0 \,, H^\varepsilon (z)e^0_m\big)_{L^2 (I_0)} = &
              \sum_{n=1 }^\infty  G_{n} (z;d_1,d_1) \big ( |(e^\varepsilon_{n} \,, e^0_m )_{L^2 (I_0)}|^2 -\delta_{nm} \big)\\
            = &\sum_{n=1 }^\infty  G_{n} (z;d_1,d_1) |(e^\varepsilon_{n} \,, e^0_m )_{L^2 (I_0)}|^2 -
              G_{m} (z;d_1, d_1)\,.
 \end{align}
 Therefore
    \begin{equation}
        \begin{split}  
       \big|\big(e^0_m \,, H^\varepsilon (z)e^0_m\big)_{L^2 (I_0)} \big|\leq &\sum_{
       \substack{n=1 \\ n \neq m}
       }^\infty   \frac{|\mathrm e^{i2d_1\sqrt {z-(\pi n/d_2)^2}}-1|}{|\sqrt{z-(\pi n/d_2)^2}|}  |(e^\varepsilon_{n} \,,\, e^0_m )_{L^2 (I_0)}|^2
       \\
       &+ \frac{|\mathrm e^{i2d_1 \sqrt {z-(\pi m/d_2)^2}}-1|}{|\sqrt{z-(\pi m/d_2)^2}|} \, \big||(e^\varepsilon_{m} \,,\, e^0_m )_{L^2 (I_0)}|^2-1  \big| \,.
        \end{split} \label{eq-1aux}
     \end{equation}
    A direct calculation shows that for $n\neq m$ we have

\begin{align} \nonumber
\big|\big(e^\varepsilon_{n} &\,,\, e^0_m \big)_{L^2 (I_0)}\big|^2
= \frac{2}{d_2}\Big(\int_{I_0} \sin \Big (\frac{\pi nx_2}{d_2} \Big)\chi_{I_\varepsilon} (x_2) \sin \Big(\frac{\pi mx_2}{d_2} \Big) \mathrm d x_2\Big)^2\\ \nonumber &=  \frac{2}{d_2}\Big(\int_{I_\varepsilon} \sin \Big(\frac{\pi nx_2}{d_2} \Big)\sin \Big(\frac{\pi mx_2}{d_2} \Big) \mathrm d x_2 \Big)^2
\\\nonumber &= \frac{2d_2}{\pi^2} \left[ \frac{\cos\left(\frac{(n-m) \pi (2t+\varepsilon)}{2d_2}\right) \sin\left(\frac{(n-m) \pi \varepsilon}{2d_2}\right)}{n-m} - \frac{\cos\left(\frac{(n+m) \pi (2t+\varepsilon)}{2d_2}\right) \sin\left(\frac{(n+m) \pi \varepsilon}{2d_2}\right)}{n+m} \right]^2
\\ \nonumber &\leq \frac{4d_2}{\pi^2} \left[ \frac{\cos^2\left(\frac{(n-m) \pi (2t+\varepsilon)}{2d_2}\right) \sin^2\left(\frac{(n-m) \pi \varepsilon}{2d_2}\right)}{(n-m)^2} + \frac{\cos^2\left(\frac{(n+m) \pi (2t+\varepsilon)}{2d_2}\right) \sin^2\left(\frac{(n+m) \pi \varepsilon}{2d_2}\right)}{(n+m)^2} \right].
\end{align}
The expressions corresponding to the above sum will be denoted as $\mathrm{I}_{n-m}$ and $\mathrm{I}_{n+m}$
Furthermore, for $z\in \mathcal B (\xi_{l,k})$ we estimate
\begin{align} \nonumber
\sum_{
       \substack{n=1 \\ n \neq m} }^\infty \frac{|\mathrm e^{i2d_1 \sqrt {z-(\pi n/d_2)^2}}-1|} {|\sqrt{z-(\pi n/d_2)^2}|}\mathrm{I}_{n-m} =&\sum_{
       \substack{n=1 \\ n \neq m} }^\infty  \frac{|\mathrm e^{i2d_1 \sqrt {z-(\pi n/d_2)^2}}-1|} {|\sqrt{z-(\pi n/d_2)^2}|}  \frac{1}{(n-m)^2}\cos^2\left(\frac{(n-m) \pi (2t+\varepsilon)}{2d_2}\right) \sin^2\left(\frac{(n-m) \pi \varepsilon}{2d_2}\right) \\ \nonumber \leq & \sum_{
       \substack{n=1 \\ n \neq m} }^\infty  \frac{|\mathrm e^{i2{d_1} \sqrt {z-(\pi n/d_2)^2}}-1|} {|\sqrt{z-(\pi n/d_2)^2}|}  \frac{1}{(n-m)^2} \sin^2\left(\frac{(n-m) \pi \varepsilon}{2d_2}\right)\\
       \nonumber
 = &\sum_{\substack{w=-m+1 \\ w \neq 0}}^\infty \frac{|\mathrm e^{i2d_1 \sqrt {z-(\pi (w+m)/d_2)^2}}-1|}{|\sqrt{z-(\pi (w+m)/d_2)^2}|}  \frac{1}{w^2} \sin^2\left(\frac{w \pi \varepsilon}{2d_2}\right)\,.
\end{align} 
Proceeding in the same way as in the proof of Lemma~\ref{le-norm_1}, we show that there exists a constant \( C \), depending on $w, m$ and  $d_2$, such that
$$
 \frac{|\mathrm e^{i2d_1 \sqrt {z-(\pi (w+m)/d_2)^2}}-1|} {|\sqrt{z-(\pi (w+m)/d_2)^2}|} \leq C\frac{1}{|w+m|} \,,\quad \mathrm{for}\quad w\geq 1-m\,,
$$
for $z\in \mathcal B (\xi_{l,k})$; see (\ref{eq-Gnest}). This enables us to establish the bound
$$\sum_{\substack{w=-m+1 \\ w \neq 0}}^\infty \frac{|\mathrm e^{i2d_1 \sqrt {z-(\pi (w+m)/d_2)^2}}-1|}{|\sqrt{z-(\pi (w+m)/d_2)^2}|}  \frac{1}{w^2} \sin^2\left(\frac{w \pi \varepsilon}{2d_2}\right)\leq
C\sum_{\substack{w=-m +1\\ w \neq 0}}^\infty \frac{1}{(w+m)}  \frac{1}{w^2} \sin^2\left(\frac{w \pi \varepsilon}{2d_2}\right)\,.
$$
The last sum can decomposed onto $\sum_{\substack{w=-m+1 }}^{-1}$ and $\sum_{\substack{w=1}}^{\infty}$.  In the first case we have the following  asymptotics
\begin{equation} \label{eq-aux12}
\sum_{w=-m +1}^{-1} \frac{1}{(w+m)w^2} \sin^2\left(\frac{w \pi \varepsilon}{2d_2}\right) = \mathcal O (\varepsilon^2) \,,
\end{equation}
which is not uniform with respect to $m$\footnote{Additional remarks are given in the following discussion.}, and in the second case, we get
\begin{equation}
\begin{split} \label{eq-auxdecomp}
    \sum_{w=1 }^{\infty} \frac{1}{(w+m)w^2} \sin^2\left(\frac{w \pi \varepsilon}{2d_2}\right) \leq &
     \sum_{w=1 }^{\infty} \frac{1}{w^3} \sin^2\left(\frac{w \pi \varepsilon}{2d_2}\right) =
    \varepsilon^3 \Big[ \sum_{w'\in\{\varepsilon, 2\varepsilon, ...\} }\frac{1}{w'^3} \sin^2\left(\frac{w' \pi }{2d_2}\right)\Big] \\  =&
    \varepsilon^3 \Big[ \sum_{w'\in\{\varepsilon, 2\varepsilon, 3\varepsilon...N_0\} }\frac{1}{w'^3} \sin^2\left(\frac{w' \pi }{2d_2}\right) + \sum_{w'\in\{N_0+\varepsilon, N_0+2\varepsilon...\} }\frac{1}{w'^3} \sin^2\left(\frac{w' \pi }{2d_2}\right)\Big]\,,
\end{split}
\end{equation}
where  $N_0$ is a positive number.
The first sum in the bracket can be bounded by the number of components $\frac{N_0}{\varepsilon}$. The second sum in the bracket can be bounded by a constant undependent of $\varepsilon$. Therefore
$$\sum_{
       \substack{n=1 \\ n \neq m} }^\infty \frac{|\mathrm e^{i2d_1 \sqrt {z-(\pi n/d_2)^2}}-1|} {|\sqrt{z-(\pi n/d_2)^2}|}\mathrm{I}_{n-m}=\mathcal O (\varepsilon^2)\,.
$$
In the same way we show
$$\sum_{
       \substack{n=1 \\ n \neq m} }^\infty \frac{|\mathrm e^{i2d_1 \sqrt {z-(\pi n/d_2)^2}}-1|} {|\sqrt{z-(\pi n/d_2)^2}|}\mathrm{I}_{n+m}=\mathcal O (\varepsilon^2)\,.
$$
Furthermore, an explicit calculus shows that
$$
|(e^\varepsilon_{m} \,,\, e^0_m )_{L^2 (I_0)}|^2-1  =\mathcal O (\varepsilon^2)\,.
$$Consequently all terms in (\ref{eq-1aux}) behave as $\mathcal O (\varepsilon^2)$.
\end{proof}
\\ \\
 The above result shows that the first term  in (\ref{eq-asypJ}) given by
$(1+o_\varepsilon^2)\big(e^0_k \,,H^\varepsilon (z(\varepsilon ) ) e^0_k\big)_{L^2 (I_0)}$ behaves as $\mathcal O (\varepsilon^2)$, since  $o_\varepsilon = o(1)$.
The next  statement  provides the  estimates for the remaining term of $J_\varepsilon$ from~(\ref{eq-asypJ}).
\begin{lemma}\label{le-norm_2}
We have
\begin{equation} \label{eq-estHmain}
    \big| \big( e^0_k\,, H^\varepsilon(z (\varepsilon )  )\eta_k ^\varepsilon (z(\varepsilon)) \,\big)_{L^2 (I_0)}\big|= \mathcal O (\epsilon^2)\,.
\end{equation}
\end{lemma}
\begin{proof}
In order to prove the above asymptotics, we write $\eta_k ^\varepsilon (z)$ by means of the basis $\{\,e_n^0\,\}_{k=1}^\infty $, i.e.
$$\eta _k^\varepsilon (z(\varepsilon)) = \sum_{\substack{n=1 \\ n \neq k} }^\infty a _n(z (\varepsilon) )e^0_n \,.$$
This, in view of (\ref{eq-Gnest}),
yields
\begin{equation} \label{eq-estim_4}
\begin{split}
     \big| \big( e^0_k \,, H^\varepsilon (z (\varepsilon ) )\eta^\varepsilon_k (z(\varepsilon))\big)_{L^2 (I_0)}\big| \leq & C
     \sum_{\substack{n=1 \\ n \neq k} }^\infty |a_n (z(\varepsilon))|\sum_{\substack{n'=1 } }^\infty \frac{1}{n'}  \Big[\,(e^\varepsilon_{n'} \,,e^0_k  )_{L^2 (I_0) }(e^0_n \,,  e^\varepsilon_{n'})_{L^2 (I_0)} -\delta_{kn'} \delta_{n'n}\Big]\\=& C(L_1+L_2)
\end{split}
\end{equation}
where
$$
L_1=\sum_{\substack{n=1\\ n\neq k}}^\infty
|a_n(z(\varepsilon))|
\sum_{\substack{n'=1\\ n'\neq n}}^\infty \frac1{n'}
\big|
 (e_{n'}^\varepsilon,e_k^0)_{L^2(I_0)}
 (e_n^0,e_{n'}^\varepsilon)_{L^2(I_0)}
\big|\,,
$$ and
$$L_2=\sum_{\substack{n=1\\ n\neq k}}^\infty
|a_n(z(\varepsilon))|
\frac1{n}
\big|
 (e_{n}^\varepsilon,e_k^0)_{L^2(I_0)}
 (e_n^0,e_{n}^\varepsilon)_{L^2(I_0)}
\big|\,.
$$
The first term
can be estimated as
\begin{equation} \label{eq-estim_4prime}
 L_1\leq \Big( \sum_{\substack{n=1 \\ n \neq k} }^\infty |a_n (z(\varepsilon))|^2  \Big)^{1/2} \cdot \sum_{\substack{n'=1 \\ n' \neq n} }^\infty\Big[ \frac{1}{n'} \big|\big(e^0_k \,,e^\varepsilon_{n'}  )_{L^2 (I_0)}\big|\cdot \big(
 \sum_{\substack{n=1 \\ n \neq k} }^\infty
     \big|(e^\varepsilon_{n'} \,, e^0_n )_{L^2 (I_0)}
     \big|^2\big)^{1/2}\Big]\,
\end{equation}
Note that
\begin{equation}\label{eq-au1}
 \big(\sum_{\substack{n=1 \\ n \neq k} }^\infty |a_n (z(\varepsilon))|^2 \,\big)^{1/2} = \| \eta^\varepsilon_{k}(z(\varepsilon))\|_{L^2 (I_0)}= \mathcal O (\varepsilon^{1/2})\,,
\end{equation}
cf.~(\ref{eq-normeta}).  Furthermore, repeating the  arguments from the proof of Lemma~\ref{le-norm} we state that
\begin{equation}\label{eq-au2}
\tilde L_1:= \sum_{\substack{n=1 \\ n \neq n'} }^\infty
     \big|(e^\varepsilon_{n'} \,, e^0_n )_{L^2 (I_0)}\big|^2 \leq C\sum_{w=-n'+1}^\infty \frac{1}{w^2}\sin ^2\Big(\frac{w\pi {\varepsilon} }{2d_2}\Big)
     =
     \mathcal O (\varepsilon)\,,
\end{equation}
where we use the reasoning analogous to \eqref{eq-aux12} and \eqref{eq-auxdecomp}, with the difference that here the relevant denominator contains the
second power instead of the third one; consequently, the resulting
asymptotics is of order $\mathcal O(\varepsilon)$ rather than
$\mathcal O(\varepsilon^2)$.
This implies
\begin{equation}\label{eq-aux9}
L_1\leq \tilde L_1^{1/2}\|\eta_k ^\varepsilon(z(\varepsilon))\|_{L^2 (I_0)}
     \cdot \Big(\sum_{\substack{n'=1 \\ n' \neq k} }^\infty \frac{1}{n'} \big|\big(e^0_k \,,e^\varepsilon_{n'}  )_{L^2 (I_0)}\big|\Big)\,.
\end{equation}
Furthermore, applying the  analysis analogous to that from the proof of Lemma~\ref{le-norm} we state
\begin{equation}\label{eq-au3}
\sum_{\substack{n'=1 \\ n' \neq k}}^\infty\frac{1}{n'} \big|\,(e^0_k \,,e^\varepsilon_{n'}  )_{L^2 (I_0)}\big|= \mathcal O (\varepsilon)\,.
\end{equation}
Combining \eqref{eq-au1}, \eqref{eq-au2}, and \eqref{eq-au3} and substituting them into
\eqref{eq-estim_4prime}, we obtain
\( L_1 = \mathcal O(\varepsilon^2) \).
The term \(L_2\) can be estimated in the same way and satisfies the same bound.
Therefore, \eqref{eq-estHmain} follows.
\end{proof}

From the above lemma we conclude  that
$ (1+o_\varepsilon)\big(e_k^0\,, H^\varepsilon(z (\varepsilon ) )\eta^\varepsilon_k (z(\varepsilon)) \big)_{L^2 (I_0)}$ appearing in (\ref{eq-asypJ})  behaves as $\mathcal O(\varepsilon^2)$.
Summarizing the statements of Lemmae~\ref{le-norm} and \ref{le-norm_2}, and inserting them into (\ref{eq-J}), we obtain, in view of (\ref{eq-asymzeta}),
\begin{equation}\label{eq-Jcorrec}
    \zeta_k (z, \varepsilon) = (1+o_\varepsilon^2)\big(e^0_k,\, K^0 (\xi_{k,l} + \delta(\varepsilon)) e_k^0\big)_{L^2 (I_0)} + J_\varepsilon, \quad \text{where} \quad J_\varepsilon = \mathcal{O}(\varepsilon^2)\,.
\end{equation}

At this stage we are ready to formulate and finish the main result of this section.
\begin{theorem} \label{th-res_2D}
Assume that $\xi_{l,k}$ is a simple eigenvalue of $-\Delta^D({\mathcal{C}})$. Then there exists a uniquely determined function $\varepsilon \mapsto z(\varepsilon)$, $\varepsilon \in ( 0, \varepsilon_0)$ such that
$
\ker K^\varepsilon \bigl(z(\varepsilon)\bigr) \neq \emptyset,
$
with the following asymptotic expansion
\begin{equation}\label{eq-main2}
z(\varepsilon) = \xi_{l,k} + \mu(\varepsilon) + i\,\nu(\varepsilon),
\end{equation}
where the real-valued functions $\mu(\varepsilon)$ and $\nu(\varepsilon)$ satisfy
\[\mu(\varepsilon) = \mathcal{O} (\varepsilon^2) \quad \text{and} \quad \nu(\varepsilon) = \mathcal{O} (\varepsilon^2)\,.
\]

Now, suppose that $\xi_{l,k}$ has multiplicity $\mathrm N$. In this case, there exist $\mathrm N$ functions $\varepsilon \mapsto z_j(\varepsilon)$, for $j = 1, \dots,\mathrm  N$, such that
$
\ker K^\varepsilon \bigl(z_j(\varepsilon)\bigr) \neq \emptyset,
$
with each function admitting the asymptotic expansion
\[
z_j(\varepsilon) = \xi_{l,k} + \mu_j(\varepsilon) + i\,\nu_j(\varepsilon),
\]
where the corresponding functions $\mu_j(\varepsilon)$ and $\nu_j(\varepsilon)$ behave as $\mathcal{O} (\varepsilon^2)$.
\end{theorem}
\begin{proof}
Assume first that $\xi_{l,k}$ is a simple eigenvalue. Then there exists a unique continuous function $\varepsilon\mapsto z(\varepsilon)$ such that
$\zeta_k (z (\varepsilon),\varepsilon) = 0$ and $z(0)=\xi_{l,k}$, cf.~Lemma~\ref{le-IFT}. Assuming now that $z=\xi_{l,k} +\delta $ we analyze the first component of (\ref{eq-Jcorrec}), namely
\begin{equation}
  \begin{split}\label{eq-aux7}
\big(  e_k^0\,, \,K^0 (\xi_{l,k}  + \delta )e_k^0 \big)_{L^2 (I_0)}
=&
 G_k (\xi_{l,k}  + \delta ;  d_1, d_1)
 \\ = &
 \frac{ e^{2d_1i \sqrt{ \xi_{l,k} + \delta - (\pi k/d_2)^2  }} - 1 }{ i \sqrt{ \xi_{l,k} + \delta - (\pi k/d_2  )^2 }} \,
= \frac{ e^{2d_1 i \sqrt{ \delta + ( \pi l/d_1)^2 }} - 1 }{ i \sqrt{ \delta + (\pi l/d_1)^2 } }
 \,,
  \end{split}
\end{equation}
where we again used the fact that $\{e^0_n\}_{n=1}^\infty $ form the orthonormal basis.

Applying now the Taylor expansion of $\delta \mapsto \sqrt{\delta + \omega^2}= \omega +\frac{1}{2\omega }\delta +\mathcal O(\delta^2)$ to
the function
$$
g(\delta) := \frac{i}{\sqrt{\delta + \omega^2}} \left( e^{2d_1 i \sqrt{\delta + \omega^2}} - 1 \right)\,,
$$
we get
$$
g(\delta) = \frac{i}{\omega} \left( e^{2d_1 i \omega} - 1 \right) + \delta \left( -\frac{d_1 e^{2d_1 i \omega}}{\omega^2} - \frac{i (e^{2d_1 i \omega} - 1)}{2\omega^3} \right) + O(\delta^2)\,.
$$

Setting
$
\omega = \frac{\pi l}{d_1}
$, we conclude from  (\ref{eq-aux7}) that
$$
\big(  e_k^0\,, \,K^0 (\xi_{l,k}  + \delta )e_k^0 \big)_{L^2 (I_0)} =  -\frac{d_1^3}{(\pi l)^2}\, \delta +\mathcal O (\delta ^2)\,.
$$
Therefore the dominating component in the above expression behaves as  $ -\frac{d_1^3}{(\pi l)^2}\, \delta $. In view of (\ref{eq-Jcorrec}) and taking into account the fact that $\zeta_k (z(\varepsilon), \varepsilon)=0$ and $J_\varepsilon = \mathcal O (\varepsilon^2)$, we conclude that
$$
\delta = \mathcal O (\varepsilon^2)\,.
$$
This completes the proof of the first part of the theorem under the assumption that $\xi_{l,k}$ is simple.

We now consider the case where $\xi_{l,k}$ is degenerate with multiplicity $\mathrm{N}$. This means that there exists $\mathrm N$ pairs of integers
$
(l_j, k_j)\,,...
$, $j=1,...,\mathrm N$ such that
\begin{equation}\label{eq-multi}
    \xi_{l,k}=\Big(\frac{l_{1}\pi }{d_1}\Big)
^{2} +\Big(\frac{k_1\pi }{d_2}\Big)^2=...=\Big(\frac{l_{\mathrm N}\pi }{d_1}\Big)
^{2}+\Big(\frac{k_{\mathrm N}\pi }{d_2}\Big)
^{2}
\end{equation}
Let $e^0_{l_j}\,e_{k_j}^0$, $j=1,...\mathrm{N}$ stand for the corresponding eigenfunctions.
Then for any $e_{k_j}^0$ we have
\begin{equation}
\zeta_{k_j } (\xi_{l,k}\,,\, 0)= \big(  e_{k_j}^0\,, \,K^0 (\xi_{l,k} )\,e_{k_j}^0 \big)_{L^2 (I_0)}\,=
 G_{k_j} (\xi_{l,k} ;  d_1, d_1)
 =
 \frac{ e^{2d_1i \sqrt{ \xi_{l,k}  - (\pi k_j /d_2)^2  }} - 1 }{ i \sqrt{ \xi_{l,k} - (\pi k_j/d_2  )^2 }}
 \,.
\end{equation}

Applying (\ref{eq-multi}), we obtain
$$
\xi_{l,k} - \left( \frac{\pi k_j}{d_2} \right)^2 = \left( \frac{\pi l_j}{d_1} \right)^2,
$$
i.e.,
$
\sqrt{ \xi_{l,k} - ( \pi k_j/d_2)^2 } = \pi l_j/d_1\,,
$
which implies that
\begin{equation} \label{eq-zerocond}
\zeta_{k_j}(\xi_{l,k},\, 0) = 0, \quad j=1,\dotsc,\mathrm{N}\,.
\end{equation}
This means that there exist \(\mathrm{N}\) functions
\[
\zeta_{k_j}(z, \varepsilon) = \big( e^{K^\varepsilon}_{k_j}(z),\, K^\varepsilon(z)\, e^{K^\varepsilon}_{k_j}(z) \big)_{L^2(I_0)}\,
\]
fulfilling (\ref{eq-zerocond}).
Employing again Lemma~\ref{le-IFT}
we conclude that  there exist $\mathrm N$ continuous  functions \(z_j(\varepsilon)\) satisfying
\[
\zeta_{k_j}(z_j(\varepsilon),\, \varepsilon) = 0.
\]
The asymptotic behavior of \(z_j(\varepsilon)\) can be studied by repeating the arguments used in the discussion of the simple pole case.
\end{proof}

\begin{remark}
    \rm{
    There  arises the question: how the components $\mu(\varepsilon)$ and $\nu(\varepsilon)$ in (\ref{eq-main2}) depend on the quantum numbers $l$ and $k$.
    It should be emphasized that the reasoning developed in Lemma~\ref{le-norm} shows that the asymptotics in equation~(\ref{eq-main2}) are not uniform with respect to $k$. This non uniformity follows, for example, from the asymptotics in equation~(\ref{eq-aux12}), which behave like $\log m$ for large $m$. Consequently, it affects (\ref{eq-Hasymp}), (\ref{eq-asypJ})
    and~(\ref{eq-main2}).    }
\end{remark}

\begin{remark} \rm{ While the leading terms of both the real and imaginary parts of the resonance pole are generically quadratic in the relative window measure, obtaining an explicit formula for the corresponding coefficient in terms of the embedded–eigenvalue eigenfunctions is substantially more delicate. In small–aperture coupling problems the leading coefficient is typically governed by the trace of the unperturbed mode on the opening. If this trace vanishes, for example  because the eigenfunction has a node on the window or due to a symmetry constraint then the leading term cancels and the resonance width is determined by higher order effects.
Therefore, if the eigenfunction corresponding to the embedded eigenvalue vanishes on the window, the $\mathcal{O}(\varepsilon^2)$ estimate is in general not expected to be optimal. A detailed analysis of this situation is beyond the scope
of the present paper.}
\end{remark}

    \section{Resonances in higher dimensional systems: A study of three-dimensional waveguides}\label{sec-3D}

\subsection{Resonances induced by a rectangular gap in a waveguide}
In this section
we study a three-dimensional waveguide characterized by width $d_2$ and height $d_3$:
$$ \Sigma:= \{(x_1, x_2, x_3 )\,:\, x_1 \in [0, \infty )\,,\,\,\,x_2 \in [\,0 , d_2] \,,\,\,\,x_3 \in [\,0 , d_3] \,  \}\,. $$
Define
$$
\bar I_\varepsilon :=\{ (d_1 , x_2, x_3)\,, \quad  x_2\in \bar{I}_{2,\varepsilon_2}=[t_2, t_2+\varepsilon_2] \,,\quad
 x_3\in \bar{I}_{3,\varepsilon_3}=[t_3, t_3+\varepsilon_3]
\}\,,
$$
where $\varepsilon_i \geq  0$, $t_i \geq 0$, $t_i+\varepsilon_i \leq d_i$ for $i=2,3$.
In the following we assume that
\begin{equation} \label{eq-rescaling}
\varepsilon = \varepsilon_2=\frac {\varepsilon_3}{a}\,,
\end{equation}
where $a$ is a positive constant. Then
$$I_\varepsilon := I_0 \setminus \bar I_\varepsilon\,,$$
where $I_0 =\{(d_1, x_2, x_3)\,:\,  \,x_2 \in [0,d_2]\,, \,x_3 \in [0,d_3] \}$. In the following we will also use $I_{j,\varepsilon_j}=I_{j,0}\setminus \bar{I}_{j,\varepsilon_j}$.

We define the Hamiltonian in the analogous way as before
$$
-\Delta^D_{I_\varepsilon} \,:\, W^{2,2}_0 (\Sigma \setminus I_\varepsilon )\to L^2 (\Sigma)\,.
$$
Its resolvent  $R(z):=(-\Delta ^D - z)^{-1}$  is now defined as an integral operator with the kernel
\begin{equation}\label{eq-reslvent_free_3D}
G (z;(x_1 , x_2, x_3), (y_1, y_2, y_3)):= \sum _{k_2=1}^\infty \sum _{k_3=1}^\infty G_{k_2, k_3} (z;x_1,y_1 ) e^0_{k_2}(x_2) e^0_{k_2} ( y_2) e^0_{k_3}(x_3) e^0_{k_3} ( y_3)\,,
\end{equation}
where \begin{equation}\label{eq-G_k_3D}
    G_{k_2,k_3} (z;x_1,y_1) = \frac{\mathrm e ^{i\tau_{k_2,k_3}(z)|x_1 + y_1|}}{i\tau_{k_2,k_3}(z)}-\frac{\mathrm e ^{i\tau_{k_2,k_3}(z)|x_1 - y_1|}}{i\tau_{k_2,k_3}(z)}\,,
\end{equation}
with  $\tau_{k_2,k_3}(z):=\sqrt{z-(\pi k_2/d_2)^2-(\pi k_3/d_3)^2}$ and $\Im \,\tau_{k_2,k_3} (z) >0$,  $k_2, k_3\in \mathbb N$.
The corresponding embedded eigenvalues take the form
\begin{equation}
    \xi_{k_1,k_2,k_3}=\Big(\frac{\pi k_1 }{d_1}\Big)^2+\Big(\frac{\pi k_2 }{d_2}\Big)^2+\Big(\frac{\pi k_3 }{d_3}\Big)^2\,.
\end{equation}
We define
$$
e^0_{n_j} = \sqrt{\frac{2}{d_j} }\sin \Big(\frac{\pi n_j x_j}{d_j}\Big)\,, \quad j=1,2,3.
$$
and for $j=2,3$ we introduce notation $e^{\varepsilon_j}_{n_j} =e^{0}_{n_j}\, \chi_{I_{\varepsilon_j}}$, and the corresponding projectors
$$P_{ n_j}^{\varepsilon_j }= (e^{\varepsilon_j}_{n_j}\,, \cdot )_{L^2 (I_{\varepsilon_j})} e^{\varepsilon_j}_{n_j}\,.
$$
Furthermore, we define the key operator
 \begin{equation} \label{eq-defK_ext_3D}
 K^\varepsilon (z) = \sum_{n_2=1}^\infty\sum_{n_3=1}^\infty G_{n_2, n_3} (z;d_1, d_1) P_{n_2}^{\varepsilon_2}P_{n_3}^{\varepsilon_3}\,:\, L^2 (I_0) \to L^2 (I_0)\,,
 \end{equation}
In analogy to the two-dimensional model, we   employ the decomposition
$$
 K^\varepsilon (z) =  K^0 (z)+ H^\varepsilon (z)\,.
$$
The main result of this section establishes asymptotics $\mathcal O (\varepsilon ^4)$ of the perturbative terms.

\begin{theorem} \label{th-res_3D}
Assume  that (\ref{eq-rescaling}) holds and  $\xi_{k_1,k_2,k_3}$ has multiplicity $\mathrm N$. Then there exist $\mathrm N$ continuous  functions $\varepsilon \mapsto z_j(\varepsilon)$, for $j = 1, \dots,\mathrm  N$, such that
$
\ker K^\varepsilon \bigl(z_j(\varepsilon)\bigr) \neq \emptyset,
$
with each function admitting the asymptotic expansion
\[
z_j(\varepsilon) = \xi_{k_1,k_2,k_3} + \mu_j(\varepsilon) + i\,\nu_j(\varepsilon),\quad  \mu_j(\varepsilon) =\mathcal{O} (\varepsilon^4)\,,\quad \nu_j(\varepsilon) =\mathcal{O} (\varepsilon^4)\,.
\]
\end{theorem}

\bigskip

The  strategy of proof  is based on the tools developed in the previous section. Here, we focus on the essential differences. First we  show the analog of Lemma~\ref{le-norm}.

\begin{lemma}
Assume that $z\in \mathcal B (\xi_{k_1, k_2, k_3})$ and assumption (\ref{eq-rescaling}) is satisfied. Then the asymptotics
\begin{equation}\label{eq_Hest_3D}
    \|H^\varepsilon (z)\|_{L^2 (I_0)\to L^2 (I_0)} = \mathcal O (\varepsilon^{1/2})\,
\end{equation}
holds.
\end{lemma}
\begin{remark}
    \rm{Let us note that the above asymptotics is weaker than the one established in Lemma~\ref{le-IFT}, due to the fact that now the volume of the gap behaves as {$|\bar I_\varepsilon|= \mathcal{O}(\varepsilon^2)$. }
    However, the asymptotics given in~(\ref{eq_Hest_3D}) is sufficient at this stage and will be strengthened in the subsequent steps of the discussion.
}
\end{remark}
\begin{proof}
In the course of this proof, all scalar products and norms without an explicit space indicator are assumed to be taken in~$L^2(I_0)$. To show the statement  we observe
\begin{equation}
\begin{split}\label{eq-auxH_3D}
|( f\,,H^\varepsilon(z)f) |\leq &  \sum_{n_2=1}^\infty\sum_{n_3=1}^\infty | G_{n_2,n_3}(z) |\,\Big[
 |(f\,,\breve{h }_{n_2,n_3}^\varepsilon f)\, (h_{n_2,n_3}^\varepsilon \,,f)|\\&+ |(f\,,h_{n_2,n_3}^\varepsilon f)\, (\breve{h}_{n_2,n_3}^\varepsilon\,,f)|
 + |(f\,,\breve{h}_{n_2,n_3}^\varepsilon )\, (\breve h_{n_2,n_3}^\varepsilon\,,f) |
\,\Big]\,.
\end{split}
\end{equation}
where  $\breve{h}_{n_2, n_3}^\varepsilon:= \breve{e}_{ n_2}^{\varepsilon} \,\breve{e}_{n_3}^{a\varepsilon}$ and
${h}_{n_2, n_3}^\varepsilon:= e^\varepsilon_{ n_2}\, {e}_{n_3}^{a\varepsilon}$, we used~(\ref{eq-rescaling}).
We assume that $f$ admits separation of variables, i.e., $f(x_2, x_3) = f_2(x_2) f_3(x_3)$. Denote  the first component of (\ref{eq-auxH_3D}) as 
$$
T =
\sum_{n_2=1}^\infty \sum_{n_3=1}^\infty  |G_{n_2,n_3}(z)|\, |(f, \breve{h}_{n_2,n_3}^\varepsilon)\,({h}_{n_2,n_3}^\varepsilon, f)|\,,
$$
which determines an analogue of (\ref{eq-defT}).
Relying on the separation of variables of the function $f$ and proceeding analogously to the proof of Lemma~\ref{le-norm}, we obtain
\begin{equation}\label{eq-CS1_3D}
\begin{split}
|(f, \breve{h}_{n_2,n_3}^\varepsilon)| &
\leq \varepsilon^{1/2} \Big( \int_{ \bar I_{2,\varepsilon}} |f_2\, \breve{e}_{n_2}^\varepsilon|^2\, \mathrm{d}x_2 \Big)^{1/2} \cdot |(f_3, \breve{e}_{n_3}^{a\varepsilon})_{L^2(I_{3,0})}| \\
& \leq \varepsilon^{1/2} \|f_2\|_{L^2(I_{2,0})} \cdot |f_{n_3}|,
\end{split}
\end{equation}
where $f_{n_3} = (e^0_{n_3}, f)_{L^2(I_{3,0})}$, and we used the estimate $|(e^{a\varepsilon}_{n_3}, f)_{L^2(I_{3,0})}| \leq |f_{n_3}|$
together  with  the Cauchy--Schwartz inequality. Similarly, we have
$$
|(h^\varepsilon_{n_2,n_3}\,,f)|\leq |f_{n_2} |\cdot|f_{n_3}|\,.
$$
Furthermore, mimicking the argument from the proof of Lemma~\ref{le-norm} yields
\begin{equation}\label{eq-Gnest_3D}
|G_{n_2,n_3} (z)|\leq C \frac{1}{\sqrt{n_2^2+n_3^2}}\,.
\end{equation}
Employing again Cauchy--Schwartz inequality and the fact that
$\sum_{n_i=1}^\infty |f_{n_i}|^2 =  \|f_i\|^2_{L^2(I_{i,0})}
$ for $i=2,3$,
we summarize
\begin{equation}
    \begin{split}
T&\leq C\varepsilon^{1/2}  \|f_2\|_{L^2(I_{2,0})}  \sum_{n_2=1}^\infty \sum_{n_3=1}^\infty \frac{1}{\sqrt{n_2^2+n_3^2}} |f_{n_2} |\cdot|f_{n_3}|^2
\\&
\leq  C\varepsilon^{1/2}  \|f_2\|_{L^2(I_{2,0})}  \Big( \sum_{n_3=1}^\infty |f_{n_3}|^2 \Big)\cdot \Big(\sum_{n_2=1}^\infty \frac{1}{n_2} |f_{n_2} |\Big)\\& \leq
 C\varepsilon^{1/2}     \|f_2\|_{L^2(I_{2,0})} \|f_3\|^2_{L^2(I_{3,0})} \Big( \sum_{n_2=1}^\infty |f_{n_2}|^2 \Big)^{1/2}\cdot \Big(\sum_{n_2=1}^\infty \frac{1}{n_2^2} \Big)^{1/2} \\&
 \leq C'\varepsilon^{1/2}  \|f_2\|^2_{L^2(I_{2,0})}  \|f_3\|^2_{L^2(I_{3,0})} = C'\varepsilon^{1/2}  \|f\|^2\,.
    \end{split}
\end{equation}
The above inequalities can be extended by the continuoity to the whole space $L^2 (I_0)$. Estimating the remaining terms of (\ref{eq-auxH_3D}) in the same way we state (\ref{eq-rescaling}).
\end{proof}

Using the above result, we can repeat the steps from the two-dimensional case and, with the help of the Continuous Implicit Function Theorem, establish the analogue of Lemma~\ref{le-IFT}. That is, there exists a uniquely defined, continuous function $\varepsilon \mapsto z(\varepsilon)$, defined for $\varepsilon \in [0, \varepsilon_0)$, $\varepsilon_0 >0$, and such that the equation
\begin{equation} \label{eq-exi_3D}
 \zeta_{k_2,k_3}(z(\varepsilon), \varepsilon) = 0
\end{equation}
is satisfied, where $\zeta_{k_2,k_3}(z, \varepsilon)$ denotes the eigenvalue of $K^\varepsilon(z)$ corresponding to the eigenfunction $h_{k_2,k_3}^{K^\varepsilon}(z)$, which represents a perturbation of $e_{k_2}^0 e_{k_3}^0$, i.e.,
$$
h_{k_2,k_3}^{K^\varepsilon}(z) =  (1+o_\varepsilon)e_{k_2}^0 e_{k_3}^0 + \eta^\varepsilon_{k_2,k_3}(z(\varepsilon))\,,
$$
where $\|\eta^\varepsilon_{k_2,k_3}(z(\varepsilon))\|_{L^2(I_0)} = \mathcal{O}(\varepsilon^{1/2})$.
To proceed further, we derive estimates analogous to~(\ref{eq-asypJ}), which in the present setting take the form
\begin{equation} \label{eq-asypJ_3D}
    J_\varepsilon = (1+ o_\varepsilon^2)\big( e_{k_2}^0 e_{k_3}^0, H^\varepsilon(z(\varepsilon))\, e_{k_2}^0 e_{k_3}^0 \big)_{L^2(I_0)} +   (1+o_\varepsilon)
    \big(e_{k_2}^0 e_{k_3}^0, H^\varepsilon(z(\varepsilon))\, \eta^\varepsilon_{k_2,k_3}(z(\varepsilon)) \big)_{L^2(I_0)}\,.
\end{equation}
The analogue of Lemma~\ref{le-norm} now takes the following form.
\begin{lemma} For $z\in \mathcal B (\xi_{k_1, k_2, k_3})$ the following asymptotics
\begin{equation}\label{eq-H3D}
\big| \left( e^0_{k_2} e^0_{k_3},  H^{\varepsilon}(z) e^0_{k_2} e^0_{k_3} ) \right)_{L^2 (I_0)} \big| = \mathcal{O}(\varepsilon^4)
\end{equation}
holds.
\end{lemma}
\begin{proof}
Similarly to before, we omit the explicit indication of the underlying space in scalar products and norms, assuming that the space is defined as
$L^2 (I_0)$. The proof is inspired by the reasoning developed in the proof of Lemma~\ref{le-norm}.
Here, we focus only on the subtle steps related to the fact that the current statement concerns the three-dimensional case.
We consider
\begin{align*}
\left( e^0_{m_2} e^0_{m_3}, H^{\varepsilon}(z) e^0_{m_2} e^0_{m_3} \right)
&\leq \sum_{n_2=1}^\infty \sum_{n_3=1}^\infty |G_{n_2 ,n_3} (z)|\,\big[| ( e^\varepsilon_{n_2}, e^{0}_{m_2} ) |^2 | ( e^{a\varepsilon}_{n_3}, e^0_{m_3} ) |^2 - \delta_{n_2 m_2} \delta_{n_3 m_3}\big]
\end{align*}
Because of that, we get:
\[
\left| \left( e^0_{m_2} e^0_{m_3}, H^{\varepsilon}(z) e^0_{m_2} e^0_{m_3} \right) \right|
\leq C \Big( \sum_{n_2=1}^\infty \sum_{n_3=1}^\infty \frac{1}{\sqrt{n_2^2 +n_3^2}} | \left( e^\varepsilon_{n_2}, e^0_{m_2} \right) |^2 | \left( e^{a\varepsilon}_{n_3}, e^0_{m_3} \right) |^2
+ \frac{1}{{\sqrt{n_2^2 +n_3^2} }}\big( | \left( e^0_{m_3}, e^0_{n_2} \right) |^2 - 1 \big) \Big)\,.
\]
In the above, we employ the analogue of equation (\ref{eq-Gnest}) adapted to the three-dimensional case.
Now, proceeding similarly as in the proof of Lemma~\ref{le-norm} we estimate
\begin{align}
| ( e^\varepsilon_{n_2}, e^0_{m_2} ) |^2
| ( e^{a\varepsilon}_{n_3}, e^0_{m_3} ) |^2
&
\leq \frac{8 d_2 d_3}{\pi^4} \big(\mathrm I_{n_2-m_2}+\mathrm I_{n_2+m_2}\big)
\cdot  \big(\mathrm I_{n_3-m_3}+\mathrm I_{n_3+m_3}\big)\,,
\end{align}

\[
{
I_{n_i \pm m_i} = \frac{1}{(n{_{i}} \pm m_i)^2} \cos^2\left( \frac{(n_i \pm m_i)\pi (2t_i + \varepsilon{_i})}{2d_i} \right) \sin^2\left( \frac{(n_i \pm m_i)\pi \varepsilon{_i}}{2d_i} \right)\,.
}
\]
Multiplication of the terms in the brackets produces fours terms. First focus on
\begin{align} \nonumber
\sum_{\substack{n_2=1 \\ n_2 \neq m_2}}^{\infty} \sum_{\substack{n_3=1 \\ n_3 \neq m_3}}^{\infty}
&|G_{n_2, n_3}(z) |\mathrm I_{n_2-m_2}\mathrm I_{n_3-m_3}
\leq \sum_{\substack{n_2=1 \\ n_2 \neq m_2}}^{\infty} \sum_{\substack{n_3=1 \\ n_3 \neq m_3}}^{\infty}
\frac{1}{\sqrt{n_2^2 + n_3^2}}
\frac{\sin^2\left( \frac{\pi(n_2 - m_2)\varepsilon}{d_2} \right)}{(n_2 - m_2)^2}
\frac{\sin^2\left( \frac{\pi(n_3 - m_3)a\varepsilon}{d_3} \right)}{(n_3 - m_3)^2}\\ \nonumber & =
\sum_{\substack{w_2=-m_2+1 \\ w_2 \neq 0}}^{\infty} \sum_{\substack{w_3=-m_3+1 \\ w_3 \neq 0}}^{\infty}
\frac{1}{\sqrt{(w_2+m_2)^2 + (w_3+m_3)^2}}\frac{1}{w_2^2{w_3^2}}
\sin^2\Big( \frac{\pi w_2\varepsilon}{d_2} \Big)
\sin^2\Big( \frac{\pi w_3a\varepsilon}{d_3} \Big)\,.
\label{eq:mainbound_SK}
\end{align}Similarly to the proof of Lemma~\ref{le-norm}, we decompose the above sums; however, this time we obtain four corresponding components, namely
\begin{equation}\label{eq-sum}
\sum_{-m_2+1}^{-1}
\sum_{-m_3+1}^{-1},\quad
\sum_{-m_2+1}^{-1}
\sum_{1}^{\infty},\quad
\sum_{1}^{\infty}
\sum_{-m_3+1}^{-1},\quad
\sum_{1}^{\infty}
\sum_{1}^{\infty}\,.
\end{equation}
where, in all double sums, the first summation refers to $w_2$, and the second to $w_3$. In the following, we denote these sums by $M_1, \dots, M_4$, respectively.
In fact $M_1$ contains $m^2$ components, each of which behaves as $\mathcal{O}(\varepsilon^4)$. Therefore, we conclude that $M_1$ admits $\mathcal{O}(\varepsilon^4)$ behavior. On the other hand, $M_2$ can be estimated using the same arguments as in the proof of Lemma~\ref{le-norm}.
Indeed the sum $\sum_{-m_2+1}^{-1}$ induces the asymptotics $\mathcal O (\varepsilon^2)$ and the sum  $
\sum_{1}^{\infty}$ yields the same asymptotics; therefore $M_2 = \mathcal O (\varepsilon^4)$. Due to the symmetry of $M_2$ and $M_3$  the same conclusion applies to $M_3$. The last sum $M_4$ in~(\ref{eq-sum}) requires a more detailed analysis. We again perform a decomposition analogous to~(\ref{eq-auxdecomp}). This yields four terms $M_{4,1}, ..., M_{4,4}$ corresponding to the decomposition of $\sum_{w_2'\in\{\varepsilon, 2\varepsilon,\dots, N_0\} }$ and
$\sum_{w_2'\in\{N_0 +\varepsilon, 2\varepsilon,  \dots\} }$,
and analogously for $w_3'$, cf.~(\ref{eq-auxdecomp}).
First we focus on
\[
M_{4, 1} := \varepsilon^6 \sum_{w_2'\in\{\varepsilon, 2\varepsilon, 3\varepsilon, \dots, N_0\} }
\sum_{w_3'\in\{\varepsilon, 2\varepsilon, 3\varepsilon, \dots, N_0\} }
\frac{1}{\sqrt{w_2'^2 + w_3'^2}} \cdot \frac{1}{w_2'^2 w_3'^2} \cdot \sin^2\left(\frac{w_2' \pi}{2d_2}\right) \cdot \sin^2\left(\frac{a w_3' \pi}{2d_3}\right).
\]
Now, we estimate the inner sum
\[
\sum_{w'_{3}\in\{\varepsilon, 2\varepsilon, 3\varepsilon, \dots, N_0\} }
\frac{1}{\sqrt{w_2'^2 + w_3'^2}} \cdot \frac{1}{w_{3}'^2 }\cdot \sin^2\left(\frac{w'_{3} \pi}{2d_{3}}\right)
\leq \frac{1}{w_{2}'} \cdot \frac{N_0}{\varepsilon}.
\]
Applying now the same reasoning to the external sum  with respect to $w_2'$,
we get
\[
M_{4,1} \leq \varepsilon^4 N_0^2\,.
\]
The further terms, for example corresponding to
 $\sum_{w_2'\in\{\varepsilon, 2\varepsilon,\dots, N_0\} }\sum_{w_3'\in\{N_0 +\varepsilon, 2\varepsilon,  \dots\} }$,

can be estimated using the same argument and the reasoning of the proof of Lemma~\ref{le-norm}.
Summarizing the above discussion, all sums in (\ref{eq-sum}) admit the asymptotics $\mathcal O (\varepsilon^4)$. The remaining expressions of the type $\mathrm I_{n_2\pm m_2}\mathrm I _{n_3\pm m_3}$, can be estimated in the same way and shows (\ref{eq-H3D}).
\end{proof}

Similarly, an analogue of Lemma~\ref{le-norm_2} follows, with obvious modifications. This consequently shows that $J_\varepsilon$ behaves as $J_\varepsilon =\mathcal O ( \varepsilon ^4)$. This completes the proof of Theorem~\ref{th-res_3D}.

\section{Final remarks and open questions} \label{sec-final}

In the discussed models, the embedded eigenvalues become resonances after introducing a small gap $\bar I_\varepsilon$ in the cavity. The asymptotic behavior of the corresponding perturbation terms is of order $\mathcal{O}(|\bar I_\varepsilon|^2)$. This means that, in the two-dimensional case, the perturbative component is of order $\varepsilon^2$, while in the three-dimensional setting, when the gap has dimensions $\varepsilon$ and $a\varepsilon$ with $a > 0$, the imaginary component behaves as
$\varepsilon^4$.
If we choose the scaling of $\bar I_\varepsilon$ such that its dimensions are $\varepsilon$ and $a$, then although this case has not been explicitly considered, relying on the same calculations we obtain resonance asymptotics of order $\mathcal{O}(\varepsilon^2)$.

The imaginary part of $z_j(\varepsilon)$ determines the width of the resonance. The inverse of $\Im z_j(\varepsilon)$ characterizes the  characteristic time scale  of the corresponding metastable state and in our case, it behaves as
$$
\tau =\mathcal O (|\bar I_\varepsilon|^{-2})\,.
$$

The discussed model, both in two and three dimensions, admits various interesting generalizations.
First, in the three-dimensional setting, the gap can be defined with a more general shape than a simple rectangle of size $[0,\varepsilon] \times [0,a\varepsilon]$. It seems reasonable to expect that if $\Omega$ represents a small gap in a two-dimensional baffle with volume $\mathrm{vol}(\Omega)$ tending to zero, then the imaginary parts of the resonance poles exhibit asymptotics of order $\mathcal{O}(\mathrm{vol}(\Omega)^2)$.

Further generalizations concern the shape of the cavity, which can be attached to the waveguide in various ways — for example, to one of the longitudinal walls. In such configurations, some of the trapped modes may correspond to embedded eigenvalues. However, in these cases, it is not necessarily true that all of them lie above the threshold of the essential spectrum. After introducing a small gap in the impenetrable  wall, some of the modes are likely to become discrete eigenvalues located below the continuum, exhibiting only a small shift in energy. The remaining modes turn into resonances. However, the models of waveguides with more complex topology raise  questions which, in the context of resonances, remain to be explored.

\section*{Data Availability Statement}
Data sharing is not applicable to this article as no new data were created or analyzed in this study.

\bigskip
\section*{Declaration of conflicting interests.} Authors declare that they have no conflicts of interest.

\section*{Acknowledgements}

 The authors are grateful to the referees for their careful reading of the manuscript,
for identifying gaps in the original arguments, and for their
 remarks and suggestions, which  improved the quality and
clarity of the paper.
This work was
partially supported by a program of the Polish Ministry of Science under the title ‘Regional
Excellence Initiative’, Project No. RID/SP/0050/2024/1.

\bigskip

\end{document}